\documentclass[a4paper,10pt]{article}
\usepackage{amsmath}
\usepackage{amsthm}
\usepackage{amsfonts}
\usepackage{amssymb}
\usepackage{graphicx}
\usepackage{color}
\usepackage{fullpage}
\usepackage{multirow}
\usepackage{subfig}
\usepackage{hyperref}
\usepackage{booktabs}
\usepackage{xspace}
\usepackage[numbers,sort&compress]{natbib}
\usepackage{epstopdf}                       
\usepackage{pstool}
\usepackage{psfrag}
\newcommand{\munit}[1]{[\mathrm{#1}]}

\newtheorem{lemma}{Lemma}
\newtheorem{theorem}{Theorem}

\newtheorem{assumption}{Assumption} 
\newtheorem{condition}{Condition}
\theoremstyle{remark}  \newtheorem{remark}{Remark}
\newcommand{\bsym}[1]{\boldsymbol{#1}}
\newcommand{\mrm}[1]{\mathrm{#1}}
\newcommand{\mc}[1]{\mathcal{#1}}

\allowdisplaybreaks

\pdfoutput = 1

\title{Observer Based Path Following for Underactuated Marine Vessels in the Presence of Ocean Currents: \\ A Local Approach \\
With proofs\thanks{The material presented in this paper has been accepted for publication in the proceedings of the IFAC World Congress 2017, Toulouse, France.}}    
\date{}
\author{M. Maghenem \thanks{University Paris-Saclay, Orsay, France. {\tt mohamed.maghenem@l2s.centralesupelec.fr}} 
\and D.J.W. Belleter \thanks{ Centre for Autonomous Marine Operations and Systems (NTNU AMOS), Department of Engineering Cybernetics, Norwegian University of Science and Technology, NO7491 Trondheim, Norway {\tt \string{dennis.belleter,claudio.paliotta,kristin.y.pettersen\string}@itk.ntnu.no}} \thanks{D.J.W. Belleter, C. Paliotta, and K.Y. Pettersen were supported by the Research Council of Norway through its Centers of Excellence funding scheme, project No. 223254 – AMOS.} 
\and C. Paliotta \footnotemark[2] \footnotemark[3] \and K.Y. Pettersen \footnotemark[2] \footnotemark[3]}
\begin{document}
\maketitle 
\begin{abstract}
In this article a solution to the problem of following a curved path in the presence of a constant unknown ocean current disturbance is presented. The path is parametrised by a path variable that is used to propagate a path-tangential reference frame. The update law for the path variable is chosen such that the motion of the path-tangential frame ensures that the vessel remains on the normal of the path-tangential reference frame. As shown in the seminal work \cite{samson1992path} such a parametrisation is only possible locally. A tube is defined in which the aforementioned parametrisation is valid and the path-following problem is solved within this tube. The size of the tube is proportional to the maximum curvature of the path. It is shown that within this tube, the closed-loop system of the proposed controller, guidance law, and the ocean current observer provides exponential stability of the path-following error dynamics. The sway velocity dynamics are analysed taking into account couplings previously overlooked in the literature, and is shown to remain bounded. Simulation results are presented.
\end{abstract}

\section{Introduction}

In this work we consider path following of an underactuated marine vessel. The problem of path following for underactuated marine vessels has its parallel in the field of mobile robotics. Therefore, a solution for $2$D path following for underactuated marine vessels based on the tools developed in \cite{samson1992path} and \cite{micaelli1993trajectory} was proposed in \citet{encarnaccao2000bpath}, where the path parametrization is used to define the path-following problem and a solution is presented using a nonlinear controller. An observer is used to incorporate the effects of the unknown, but constant ocean current. Part of the state is shown to be stable and the zero dynamics are analysed and shown to be well behaved. However, this is done under the assumption that the total speed is constant. This requires active control of the forward velocity to cancel the effect of the sideways velocity induced by turning. Moreover, the parametrisation from \citet{micaelli1993trajectory} is only valid locally, making the path-following result valid locally. Another local result based on the same parametrisation is obtained in \citet{do2004state}. In this work a practical stability result is shown for the path-following states of an underactuated surface vessel in the presence of an environmental disturbance. However, in \cite{do2004state} there is a problem with the bound of the practical stability result and the error cannot be made arbitrary small. 
Moreover, a simplified model with diagonal system matrices is used and the interconnection between the total velocity and the sideways velocity is not taken into account in the analysis of the zero dynamics. 

In \citet{lapierre2007nonlinear} the path parametrisation is used to solve the path-following problem globally. This is done using another result, first described for the control of mobile robots in \citet{soetanto2003adaptive}. In particular, it is achieved by adapting the parametrisation of the path in order to avoid singularities in the parametrisation of the path. The work in \cite{lapierre2007nonlinear} does not consider environmental disturbances, however. It focuses on stabilisation of the path-following states but does not analyse the zero dynamics. A similar approach is taken in \citet{borhaug2006path} in which the frame is propagated differently. 
In \citet{borhaug2006path} a look-ahead based steering law is used to guide the vehicle to the path. Stability of the path-following errors is shown using cascaded systems theory, and the zero dynamics are analysed and shown to be well behaved. To take into account ocean currents, the work in \citet{borhaug2006path} is extended in \citet{Borhaug2008} by adding integral action to the steering laws. However, the results in \citet{Borhaug2008} are only valid for straight-line path following. The work in \citet{Borhaug2008} was revisited in \citet{caharija2012b} for surface vessels using a relative velocity model. Experimental results were added in \citet{caharija2016integral}. 

To address curved paths, the work of \citet{borhaug2006path} is extended with an ocean current observer in \citet{moe2014path}. However, in \citet{moe2014path} the zero dynamics are not analysed and the suggested input signals contain the unknown ocean current. Another line-of-sight (LOS) guidance approach for path following is presented in \citet{Fossen2003}, which is used to follow a path made of straight-line sections connecting way points. These concepts are further developed to circles in \citet{breivik2004path} where the vessel is regulated to the tangent of its projection on the circle. These works do not consider environmental disturbances.

This paper considers path-following of underactuated marine vessels in the presence of constant ocean currents, for general paths. A line-of-sight guidance law, an ocean current observer, and a local parametrisation of the path are used. Compared to \citet{do2004state}, in this work the parametrisation is adapted to include the effect of the unknown ocean currents and a complete analysis of the sway velocity dynamics are given, taking into account the coupling between the total velocity and the sway velocity. Moreover, the mass and damping matrix are allowed to be non-diagonal and we avoid the problems with the practical stability result.
Due to the locality of the parametrisation it can only be used in a certain tube around the path whose size depends on the maximum curvature of the path. When in this tube, it is shown that the closed-loop system of the controllers and the ocean current observer provides exponential stability of the path-following error dynamics.

The article is organized as follows: in Section \ref{COL-sec:mdl} the vessel model and the problem definition are presented. The path parametrisation is introduced in Section \ref{COL-sec:pd}. Section \ref{COL-sec:crtl} presents the ocean current observer, the guidance law, and controllers. The closed-loop system is then formulated and analysed in Section \ref{COL-sec:clsys}. A simulation case study is presented in Section \ref{COL-sec:case} and conclusions are given in Section \ref{COL-sec:cncl}.

\section{Vessel Model} \label{COL-sec:mdl}
In this section we consider the model which can be used to describe an autonomous surface vessel or an autonomous underwater vehicle moving in a plane. Recall, that the model can be represented in component form as
\begin{subequations} \label{COL-eq:dynsys}
\begin{align}  
\dot{x}&=u_r \cos(\psi) - v_r \sin(\psi) + V_x, \label{COL-eq:xdot} \\
\dot{y}&=u_r \sin(\psi) + v_r \cos(\psi) + V_y, \label{COL-eq:ydot} \\
\dot{\psi} &= r, \label{COL-eq:phidot} \\
\dot{u}_r &= F_{u_r}(v_r,r)-\tfrac{d_{11}}{m_{11}}u_r + \tau_u, \label{COL-eq:urdot}\\
\dot{v}_r &= X(u_r)r+Y(u_r)v_r, \label{COL-eq:vrdot} \\
\dot{r} &= F_r(u_r,v_r,r)+\tau_r, \label{COL-eq:rdot}
\end{align}
\end{subequations}
The functions $X(u_r)$, $Y(u_r)$, $F_u$, and $F_r$ are given by 
\begin{subequations}
\begin{align}
F_{u_r}(v_r,r) &\triangleq  \frac{1}{m_{11}}(m_{22}v_r+m_{23}r)r, \\ 
X(u_r) &\triangleq \frac{m^2_{23}-m_{11}m_{33}}{m_{22}m_{33}-m^2_{23}}u_r + \frac{d_{33}m_{23}-d_{23}m_{33}}{m_{22}m_{33}-m^2_{23}},\\ \label{COL-eq:Xur}
Y(u_r) &\triangleq  \frac{(m_{22}-m_{11})m_{23}}{m_{22}m_{33}-m^2_{23}}u_r - \frac{d_{22}m_{33}-d_{32}m_{23}}{m_{22}m_{33}-m^2_{23}},\\ \label{COL-eq:Yur}
\begin{split}
F_r(u_r,v_r,r)&\triangleq  \frac{m_{23}d_{22}-m_{22}(d_{32}+(m_{22}-m_{11})u_r)}{m_{22}m_{33}-m^2_{23}}v_r \\
&\quad+ \frac{m_{23}(d_{23}+m_{11}u_r)-m_{22}(d_{33}+m_{23}u_r)}{m_{22}m_{33}-m^2_{23}}r.
\end{split}
\end{align}
\end{subequations}
Note that the functions $X(u_r)$ and $Y(u_r)$ are linear functions of the velocity. The kinematic variables are illustrated in Figure \ref{COL-fig:defin}. The ocean current satisfies the following assumption.
\begin{assumption} \label{COL-assum:current}
The ocean current is assumed to be constant and irrotational with respect to the inertial frame, i.e. $\bsym{V}_c\triangleq [V_x,V_y,0]^T$. Furthermore, it is bounded by $V_{\max}>0$ such that $\|\bsym{V}_{c}\|=\sqrt{V^2_{x}+V^2_{y}}\leq V_{\max}$.
\end{assumption}
Moreover, for the considered range of values of the desired surge velocity $u_{rd}$ the following assumption holds.
\begin{assumption} \label{COL-assum:Yur}
It is assumed that $Y(u_r)$ satisfies 
$
Y(u_r) \leq -Y_{\min}< 0, \,\forall u_r \in [-V_{\max},u_{rd}],
$
i.e. $Y(u_r)$ is negative for the range of desired velocities considered.
\end{assumption} 
\begin{remark}
Assumptions \ref{COL-assum:Yur} is satisfied for commercial vessels by design, since the converse would imply an undamped or nominally unstable vessel in sway.
\end{remark}

Additionally we assume that the following assumption holds
\begin{assumption} \label{COL-assum:vel}
It is assumed that $2V_{\max} < u_{rd}(t)~\forall t$, i.e. the desired relative velocity of the vessel is larger than the maximum value of the ocean current.
\end{assumption} 
Assumption \ref{COL-assum:vel} assures that the vessel has enough propulsion power to overcome the ocean current affecting it. The factor two in Assumption \ref{COL-assum:vel} adds some extra conservativeness to bound the solutions of the ocean current observer, this is discussed further in Section \ref{COL-sec:cncl}.

\begin{figure}[!htb]
\centering
\includegraphics[width=.5\columnwidth]{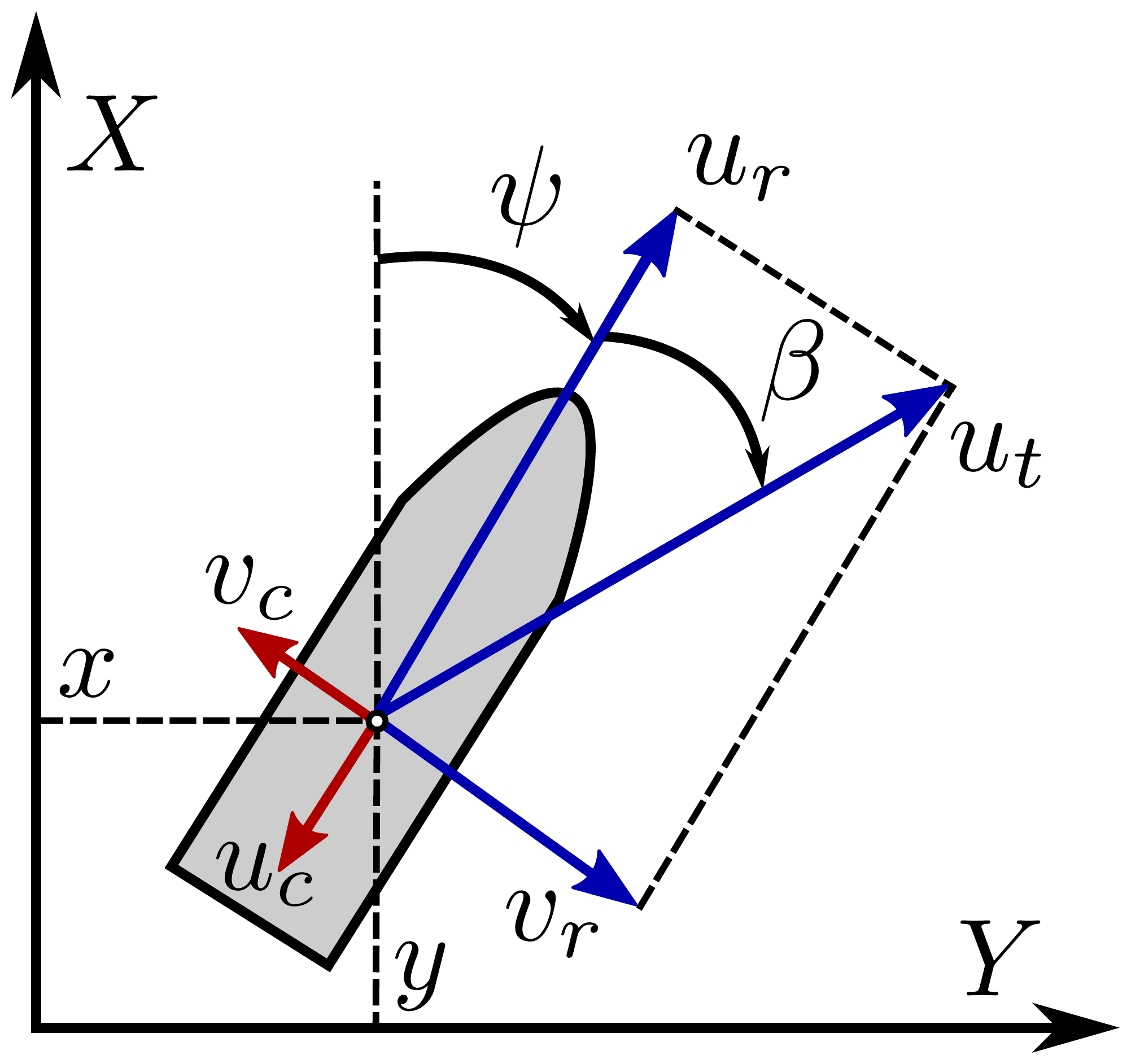}
\caption{Definition of the ship's kinematic variables.}\label{COL-fig:defin}
\end{figure}

\section{Problem definition} \label{COL-sec:pd}
The goal is to follow a smooth path $P$, parametrised by a path variable $\theta$, by appropriately controlling the ship's surge velocity and yaw rate. For an underactuated vessel, path following can be achieved by positioning the vessel on the path with the total velocity $u_t \triangleq \sqrt{u^2_r+v^2_r}$ (see Figure \ref{COL-fig:defin}) tangential to the path. To express the path-following error we propagate a path-tangential frame along $P$ such that the vessel will be on the normal of the path-tangential frame at all time. This is illustrated in Figure \ref{COL-fig:path}. The preceding implies that the progression of the path-tangential frame is controlled such that the path-following error takes the form:
\begin{align} \label{COL-eq:ye}
\begin{bmatrix} x_{b/p} \\ y_{b/p} \end{bmatrix} &= \begin{bmatrix} \cos(\gamma_p(\theta)) & \sin(\gamma_p(\theta)) \\ -\sin(\gamma_p(\theta)) & \cos(\gamma_p(\theta)) \end{bmatrix} \begin{bmatrix} x - x_{P}(\theta) \\ y - y_{P}(\theta) \end{bmatrix} = \begin{bmatrix} 0 \\ y_{b/p} \end{bmatrix}, 
\end{align}
where $\gamma(\theta)$ is the angle of the path with respect to the $X$-axis, $x_{b/p}$ is the deviation from the normal in tangential direction, and $y_{b/p}$ is the deviation from the tangent in normal direction. The time derivative of the angle $\gamma(\theta)$ is given by $\dot{\gamma}(\theta) = \kappa(\theta)\dot{\theta}$ where $\kappa(\theta)$ is the curvature of $P$ at $\theta$. The goal is to regulate $x_{b/p}$ and $y_{b/p}$ to zero.

\subsection{Locally valid parametrisation}
The error in the tangential direction $x_{b/p}$ will be kept at zero by the choice of the update law for the path variable $\theta$, i.e. the vehicle is kept on the normal. It is well known that such a parametrisation will only be unique locally \cite{samson1992path}. In particular, such a unique expression exists when the vehicle is closer to the path than the inverse of the maximum curvature of the path, i.e. when $y_{b/p} < 1/\kappa_{\max}$ where $\kappa_{\max}$ is the maximum curvature of the path. Note that this is equivalent to being closer than the radius of the smallest inscribed circle of the path. To design such a parametrisation we first consider the error dynamics of the vessel with respect to the path frame, which is given by:
\begin{subequations} \label{COL-eq:frame}
\begin{align} 
\dot{x}_{b/p} &= - \dot{\theta} (1 - \kappa(\theta)y_{b/p}) + u_t \cos(\chi - \gamma_p(\theta)) + V_T, \label{COL-eq:framea}\\ 
\dot{y}_{b/p} &= u_t \sin(\chi - \gamma_p(\theta)) + V_N - \kappa(\theta)\dot{\theta}x_{b/p}, \label{COL-eq:frameb}
\end{align}
\end{subequations}
where $\chi \triangleq \psi + \beta$ is the course angle (see Figure \ref{COL-fig:defin}) and $V_T \triangleq V_x \cos(\gamma_p(\theta)) + V_y \sin(\gamma_p(\theta))$ and $V_N \triangleq V_y \cos(\gamma_p(\theta)) - V_x \sin(\gamma_p(\theta))$ are the ocean current component in the tangential direction and normal direction of the path-tangential reference frame, respectively. Consequently, if the path variable $\theta$ is updated according to
\begin{equation} \label{COL-eq:theta}
\dot{\theta} = \frac{u_{t}\cos\left(\chi - \gamma_p(\theta)\right) + V_T}{1 - \kappa(\theta)y_{b/p}},
\end{equation}
the vessel stays on the normal when it starts on the normal. In particular, substitution of \eqref{COL-eq:theta} in \eqref{COL-eq:framea} results in $\dot{x}_{b/p}=0$. To make sure that the update law \eqref{COL-eq:theta} is well defined the following condition should be satisfied
\begin{condition} \label{COL-cond:curv}
To have a well defined update law for the path variable $\theta$ it should hold that
$
1 - \kappa(\theta)y_{b/p} \neq 0
$
for all time.
\end{condition}
Note that Condition \ref{COL-cond:curv} implies that the update law is well defined within the tube of radius $y_{b/p} < 1/\kappa_{\max}$ which results in the parametrisation being only locally valid.

\begin{figure}[!htb]
\centering
\includegraphics[width=.7\columnwidth]{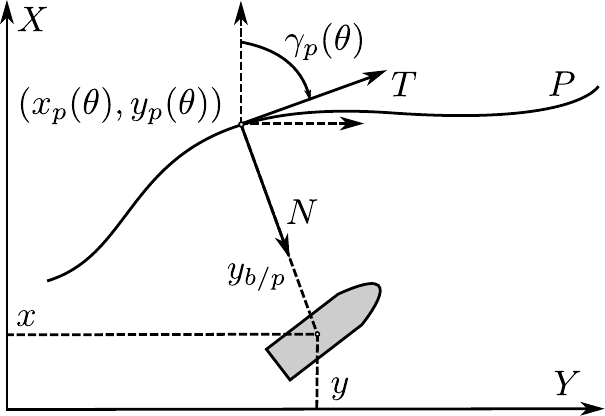}
\caption{Definition of the path.}\label{COL-fig:path}
\end{figure}

The update law \eqref{COL-eq:theta} depends on the current component $V_T$. However, since the current is assumed unknown we have to replace $V_T$ by its estimate $\hat{V}_T \triangleq \hat{V}_x \cos(\gamma(\theta)) + \hat{V}_y \sin(\gamma(\theta))$. Consequently, the last equality of \eqref{COL-eq:ye} does not hold until the current is estimated correctly. Therefore, \eqref{COL-eq:ye} takes the form
\begin{equation} \label{COL-eq:delta_x}
\begin{bmatrix} x_{b/p} \\ y_{b/p} \end{bmatrix} = \begin{bmatrix} \cos(\gamma(\theta)) & \sin(\gamma(\theta)) \\ -\sin(\gamma(\theta)) & \cos(\gamma(\theta)) \end{bmatrix} \begin{bmatrix} x - x_{P}(\theta) \\ y - y_{P}(\theta) \end{bmatrix}.
\end{equation}
To force \eqref{COL-eq:delta_x} to become equal to \eqref{COL-eq:ye} once the ocean current is estimated correctly we augment \eqref{COL-eq:theta} to be
\begin{equation} \label{COL-eq:theta2}
\dot{\theta} = \frac{u_{t}\cos\left(\chi - \gamma_p(\theta)\right) + \hat{V}_T + k_\delta x_{b/p}}{1 - \kappa(\theta)y_{b/p}},
\end{equation}
such that the path-tangential reference frame propagates based on an estimate of the ocean current and has a restoring term to drive $x_{b/p}$ to zero. Hence, substituting \eqref{COL-eq:theta2} in \eqref{COL-eq:framea} gives
\begin{equation}
\dot{x}_{b/p} = -k_\delta x_{b/p} + \tilde{V}_T, \label{COL-eq:ddeltax}\\
\end{equation}
which shows that if the estimate of the current has converged the restoring term $k_\delta x_{b/p}$ remains to drive $x_{b/p}$ to zero after which the vessel remains on the normal of the path-tangential frame.

The dynamics of the error along the normal are given by
\begin{align}
\dot{y}_{b/p} &= u_{t}\sin(\chi - \gamma_p(\theta)) + V_N - x_{b/p}\kappa(\theta)\dot{\theta}. \label{COL-eq:dye}
\end{align}
In the next section a guidance law is chosen to stabilise the origin of the dynamics \eqref{COL-eq:ddeltax}-\eqref{COL-eq:dye} and achieve the goal of path following. 

Note that since the path parametrisation is only local, we can only utilise it within a tube around the path with radius $1/\kappa_{\max}$. To achieve global results this tube needs to be made attractive and invariant, such that the vehicle first converges to the tube after which the unique parametrisation to achieve path-following can be used. The disadvantage of this is that a two-step approach is needed to solve the path-following problem, which complicates the analysis. There is, however, also a big advantage to this approach, since extra design freedom is available when making the tube attractive. This allows one to design the approach behaviour and convergence when far from the path, while for a global one-step approach this is in general not possible to do independently of the behaviour close to the path. Hence, for the one-step approach the global behaviour will be a compromise between the desired behaviour far away from the path and the desired behaviour close to the path. For the two-step approach, the behaviour far away from the path and close to the path can be optimised independently. This, for instance, allows strategies where the vehicle moves along the normal of the path to reach the path as fast as possible. Moreover, in cluttered environments this allows the vessel to converge to the path along a clearly defined approach path, after which it can switch to the guidance strategy that allows it to follow the desired path $P$.  

\section{Controller, Observer, and Guidance}  \label{COL-sec:crtl}
In this section we design the two control laws $\tau_u$ and $\tau_r$, and the ocean current estimator that are used to achieve path-following. In the first subsection we present the velocity control law $\tau_u$. The second subsection presents the ocean current estimator. The third subsection first presents the guidance to be used within the tube. 

\subsection{Surge velocity control}
The velocity control law is a feedback-linearising P-controller that is used to drive the relative surge velocity to a desired $u_{rd}$ and is given by
\begin{equation} \label{COL-eq:tauu}
\tau_u = -F_{u_r}(v_r,r) + \dot{u}_{rd} + \frac{d_{11}}{m_{11}}u_{rd} - k_u (u_r - u_{rd}),
\end{equation}
where $k_u>0$ is a constant controller gain. It is straightforward to verify that \eqref{COL-eq:tauu} ensures global exponential tracking of the desired velocity. In particular, when \eqref{COL-eq:tauu} is substituted in \eqref{COL-eq:urdot} we obtain
\begin{equation} \label{COL-eq:utild}
\dot{\tilde{u}}_r = -k_{u}(u_r-u_{rd}) = -k_u \tilde{u}_r,
\end{equation}
where $\tilde{u}_r \triangleq u_r -u_{rd}$. Consequently, the velocity error dynamics are described by a stable linear systems, which assures exponential tracking of the desired velocity $u_{rd}$.

\subsection{Ocean current estimator} \label{COL-subsec:obs}
This subsection presents the ocean current estimator introduced in \cite{aguiar2007dynamic}. This observer provides the estimate of the ocean current needed to implement \eqref{COL-eq:theta2} and the guidance law developed in the next subsection. Rather than estimating the time-varying current components in the path frame $V_T$ and $V_N$ the observer is used to estimate the constant ocean current components in the inertial frame $V_x$ and $V_y$. The observer from \cite{aguiar2007dynamic} is based on the kinematic equations of the vehicle, i.e. \eqref{COL-eq:xdot} and \eqref{COL-eq:ydot}, and requires measurements of the vehicle's $x$ and $y$ position in the inertial frame. The observer is formulated as
\begin{subequations} \label{COL-eq:obs}
\begin{align}
\dot{\hat{x}} &= u_r\cos(\psi) - v_r\sin(\psi) + \hat{V}_x + k_{x_1} \tilde{x} \\
\dot{\hat{y}} &= u_r\sin(\psi) + v_r\cos(\psi) + \hat{V}_y + k_{y_1} \tilde{y} \\
\dot{\hat{V}}_x &= k_{x_2}\tilde{x} \\
\dot{\hat{V}}_y &= k_{y_2}\tilde{y}
\end{align}
\end{subequations}
where $\tilde{x} \triangleq x - \hat{x}$ and $\tilde{y} = y -\hat{y}$ are the positional errors and $k_{x_1}$, $k_{x_2}$, $k_{y_1}$, and $k_{y_2}$ are constant positive gains. Consequently, the estimation error dynamics are given by
\begin{equation} \label{COL-eq:obserr}
\begin{bmatrix} \dot{\tilde{x}} \\ \dot{\tilde{y}} \\ \dot{\tilde{V}}_x \\ \dot{\tilde{V}}_y \end{bmatrix} = \begin{bmatrix} -k_{x_1} & 0 & 1 & 0 \\ 0 & -k_{y_1} & 0 & 1 \\ -k_{x_2} & 0 & 0 & 0 \\ 0 & -k_{y_2} & 0 & 0 \end{bmatrix} \begin{bmatrix} \tilde{x} \\ \tilde{y} \\ \tilde{V}_x \\ \tilde{V}_y \end{bmatrix}.
\end{equation}
which is a linear system with negative eigenvalues. Hence, the observer error dynamics are globally exponentially stable at the origin. Note that this implies that also $\hat{V}_T$ and $\hat{V}_N$ go to $V_T$ and $V_N$ respectively with exponential convergence since it holds that
\begin{subequations}
\begin{align}
\hat{V}_T &= \hat{V}_x \cos(\gamma(\theta)) + \hat{V}_y \sin(\gamma(\theta)), \\
\hat{V}_N &= -\hat{V}_x \sin(\gamma(\theta)) + \hat{V}_y \cos(\gamma(\theta)).
\end{align}
\end{subequations}

For implementation of the controllers it is desired that $\Vert \hat{V}_N(t) \Vert < u_{rd}(t)~\forall t$. To achieve this we first choose the initial conditions of the estimator as 
$
[\hat{x}(t_0),\hat{y}(t_0),\hat{V}_x(t_0),\hat{V}_y(t_0)]^T = [x(t_0),y(t_0),0,0]^T.
$
Consequently, the initial estimation error is given by 
$
[ \tilde{x}(t_0),\tilde{y}(t_0),\tilde{V}_x(t_0),\tilde{V}_y(t_0)]^T = [0,0,V_x,V_y]^T,
$
which has a norm smaller than or equal to $V_{\max}$ according to Assumption \ref{COL-assum:current}. Now consider the function
\begin{equation}
W(t) = \tilde{x}^2 + \tilde{y}^2 + \frac{1}{k_{x_2}}\tilde{V}^2_x + \frac{1}{k_{y_2}}\tilde{V}^2_y,
\end{equation}
which has the following time derivative
\begin{align}
\begin{split}
\dot{W}(t) &= -2k_{x_1}\tilde{x}^2 - 2k_{y_1}\tilde{y}^2 \leq 0.
\end{split}
\end{align}
This implies that $W(t) \leq \Vert W(t_0) \Vert$. From our choice of initial conditions we know that
\begin{equation}
\Vert W(t_0) \Vert = \frac{1}{k_{x_2}}V^2_x + \frac{1}{k_{y_2}}V^2_y \leq \frac{1}{\min(k_{x_2},k_{y_2})}V^2_{\max}.
\end{equation}
Moreover, it is straightforward to verify
\begin{equation}
\frac{1}{\max(k_{x_2},k_{y_2})}\Vert \tilde{\bsym{V}}_c(t) \Vert^2 \leq W(t).
\end{equation}
Combining the observations given above we obtain
\begin{equation}
\frac{1}{\max(k_{x_2},k_{y_2})}\Vert \tilde{\bsym{V}}_c(t) \Vert^2 \leq \frac{1}{\min(k_{x_2},k_{y_2})}V^2_{\max}.
\end{equation}
Consequently, we obtain
\begin{equation}
\Vert \tilde{\bsym{V}}_c(t) \Vert \leq \sqrt{\frac{\max(k_{x_2},k_{y_2})}{\min(k_{x_2},k_{y_2})}}V_{\max} < \sqrt{\frac{\max(k_{x_2},k_{y_2})}{\min(k_{x_2},k_{y_2})}}u_{rd}(t),~\forall t,
\end{equation}
which implies that if the gains are chosen as $k_{x_2}=k_{y_2}$ we have 
\begin{equation}
\Vert \hat{V}_N \Vert \leq 2 V_{\max} \leq u_{rd}(t),~\forall t.
\end{equation}
Hence, $\Vert \hat{V}_N \Vert < u_{rd}(t),~\forall t$ if $2 V_{\max} < u_{rd}(t),~\forall t$.

\begin{remark}
The bound $2V_{\max} < u_{rd},~\forall t$, is only required when deriving the bound on the solutions of the observer. In particular, it is required to guarantee that $\Vert \hat{V}_N \Vert < u_{rd}(t),~\forall t$. For the rest of the analysis it suffices that $V_{\max} < u_{rd},~\forall t$. Therefore, if the more conservative bound $2V_{\max} < u_{rd},~\forall t$, is not satisfied the observer can be changed to an observer that allows explicit bounds on the estimate $\hat{V}_N$, e.g. the observer developed \citet{narendra1987new}, rather than an observer that only provides a bound on the error $\tilde{\bsym{V}}_c$ as is the case here. For practical purposes the estimate can also be saturated such that $\Vert \hat{V}_N \Vert < u_{rd},~\forall t$, which is the approach taken in \citet{moe2014path}. However, in the theoretical analysis of the yaw controller we use derivatives of $\hat{V}_{N}$ which will be discontinuous when saturation is applied.
\end{remark}

\subsection{Guidance}

This subsection presents the guidance that is used in combination with the local parametrisation. Since, the chosen parametrisation is only valid in a tube around the path, the proposed guidance is designed for operation in the tube. Inside the tube we propose the following guidance law
\begin{equation} \label{COL-eq:psi}
\psi_d = \gamma(\theta) - \mathrm{atan}\left(\frac{v_r}{u_{rd}}\right) - \mathrm{atan}\left(\frac{y_{b/p} + g}{\Delta}\right).
\end{equation} 
The guidance law consists of three terms. The first term is a feedforward of the angle of the path with respect to the inertial frame. The second part is the desired side-slip angle, i.e. the angle between the surge velocity and the total speed when $u_r \equiv u_{rd}$. This side-slip angle is used to make the vehicle's total speed tangential to the path when the sway velocity is non-zero. The third term is a line-of-sight (LOS) term that is intended to steer the vessel to the path, where $g$ is a term dependent on the ocean current. The choice of $g$ provides extra design freedom to compensate for the component of the ocean current along the normal axis $V_N$. To analyse the effect of this guidance law and to design $g$ we consider the error dynamics along the normal \eqref{COL-eq:dye}. To do this we substitute \eqref{COL-eq:psi} in \eqref{COL-eq:dye} and obtain
\begin{subequations} \label{COL-eq:dye2}
\begin{align}
\dot{y}_{b/p} &= u_{td} \sin\left(\psi_d +\tilde{\psi} +\beta_d - \gamma_p(\theta)\right) + V_N - x_{b/p}\kappa(\theta)\dot{\theta} + \tilde{u}_r \sin(\psi - \gamma_p(\theta)) \\
&= -u_{td} \frac{y_{b/p} + g}{\sqrt{(y_{b/p}+g)^2 + \Delta^2}} + V_N + G_1(\tilde{\psi},\tilde{u}_r,x_{b/p},\psi_d,y_{b/p},u_{td},\dot{\gamma}_p(\theta))
\end{align}
\end{subequations}
where $G_1(\cdot)$ is a perturbing term given by
\begin{align} \label{COL-eqeq:G}
\begin{split}
G_1(\cdot ) = &~u_{td}\left[ 1- \cos(\tilde{\psi}) \right]\sin \left( \arctan \left(  \frac{y_{b/p} + g}{\Delta} \right) \right) + \tilde{u}_r \sin(\psi-\gamma_p(\theta))  \\  &+ u_{td}\cos \left( \arctan \left(  \frac{y_{b/p} + g}{\Delta} \right) \right) \sin (\tilde{\psi})  - x_{b/p} \dot{\gamma}_p(\theta)
\end{split}
\end{align} 
and $u_{td} \triangleq \sqrt{u^2_{rd} + v^2_r}$ is the desired total velocity. Note that $G_1(\cdot)$ satisfies
\begin{subequations} \label{COL-eq:Gbnd}
\begin{align}
G_1(0,0,0,\psi_d,y_{b/p},u_{td},\dot{\gamma}_p(\theta)) &= 0 \\
\Vert G_1(\tilde{\psi},\tilde{u}_r,x_{b/p},\psi_d,y_{b/p},u_{td},\dot{\gamma}_p(\theta)) \Vert &\leq \zeta(\dot{\gamma}_p(\theta),u_{td})\Vert (\tilde{\psi},\tilde{u},x_{b/p}) \Vert,
\end{align}
\end{subequations}
where $\zeta(\dot{\gamma}_p(\theta),u_{td}) > 0 $, which shows that $G_1(\cdot)$ is zero when the perturbing variables are zero and that it has maximal linear growth in the perturbing variables. 

To compensate for the ocean current component $V_N$ the variable $g$ is now chosen to satisfy the equality
\begin{equation}
u_{td} \frac{g}{\sqrt{\Delta^2+(y_{b/p} +g)^2}} = \hat{V}_N.
\end{equation}
which is a choice inspired by \cite{moe2014path}.
In order for $g$ to satisfy the equality above, $g$ should be the solution of the following second order equality
\begin{equation}
\underbrace{(u^2_{td} - \hat{V}^2_N)}_{-a}\left(\frac{g}{\hat{V}_N}\right)^2 = \underbrace{\Delta^2 + y^2_{b/p}}_c + 2 \underbrace{y_{b/p} \hat{V}_N}_b \left(\frac{g}{\hat{V}_N}\right),
\end{equation}
hence we choose $g$ to be
\begin{equation}
g = \hat{V}_{N} \frac{b + \sqrt{b^2 - ac}}{-a},
\end{equation}
which has the same sign as $\hat{V}_N$ and is well defined for $(u^2_{rd} - \hat{V}^2_N) > 0$. Moreover, since
\begin{equation}
\sqrt{b^2 - ac} = \sqrt{\Delta^2(u^2_{td}-\hat{V}^2_N)+y^2_{b/p}u^2_{td}}
\end{equation}
solutions are real for $(u^2_{rd} - \hat{V}^2_N) > 0$.

Consequently if we substitute this choice for $g$ in \eqref{COL-eq:dye2} we obtain
\begin{equation} \label{COL-eq:dye3}
\dot{y}_{b/p} = -u_{td} \frac{y_{b/p}}{\sqrt{(y_{b/p}+g)^2 + \Delta^2}} + \tilde{V}_N + G_1(\tilde{\psi},\tilde{u},x_{b/p},\psi_d,y_{b/p},u_{td},\dot{\gamma}_p(\theta)).
\end{equation}

The desired yaw rate can be found by taking the time derivative of \eqref{COL-eq:psi} resulting in
\begin{equation} \label{COL-eq:dpsid}
\dot{\psi}_d = \kappa(\theta)\dot{\theta} + \frac{\dot{v}_r u_{rd} -  \dot{u}_{rd} v_r}{u^2_{rd} + v^2_r} + \frac{\Delta(\dot{y}_{b/p} + \dot{g})}{\Delta^2 + (y_{b/p} + g)^2},
\end{equation}
where $\dot{v}_r$ as given in \eqref{COL-eq:vrdot}, $\dot{y}_{b/p}$ in \eqref{COL-eq:dye3}, and $\dot{g}$ is given by
\begin{equation}
\dot{g} = \dot{\hat{V}}_N \frac{b+ \sqrt{b^2-ac} }{-a} + \frac{\partial g}{\partial a} \dot{a} +  \frac{\partial g}{\partial b} \dot{b} +  \frac{\partial g}{\partial c} \dot{c},
\end{equation}
where
\begin{subequations}
\begin{gather}
  \frac{\partial g}{\partial a} = \hat{V}_N \frac{c}{2 a \sqrt{b^2 - ac}} + \hat{V}_N \frac{b+\sqrt{b^2-ac}}{a^2}, \\
  \dot{a} = 2 \hat{V}_N \dot{\hat{V}}_N - 2 u_{rd} \dot{u}_{rd} - 2 v_r \left[ X(u_r) r + Y(u_r) v_r \right], \\
   \frac{\partial g}{\partial b} = \hat{V}_N \frac{b+\sqrt{b^2-ac}}{a\sqrt{b^2 - ac}}, \\
  \dot{b} = 2 \hat{V}_N \dot{y}_{b/p} + 2 \dot{\hat{V}}_N y_{b/p}, \qquad \frac{\partial g}{\partial c} = \hat{V}_N \frac{1}{2 \sqrt{b^2 - ac}}, \qquad \dot{c} = 2 y_{b/p} \dot{y}_{b/p}.
\end{gather}
\end{subequations}
Note that $\dot{y}_{b/p}$ appears a number of times in the expression for $\dot{\psi}_d$ and that $\dot{y}_{b/p}$ depends on $\tilde{V}_N$. Consequently, $\dot{\psi}_d$ depends on an unknown variable and cannot be used to control the yaw rate. This was not considered in \cite{moe2014path} where the proposed controller contained both $\dot{\psi}_d$ and $\ddot{\psi}_d$.

Moreover, since $\dot{\psi}_d$ contains $\dot{v}_r$, which depends on $r = \dot{\psi}$, the yaw rate error $\dot{\tilde{\psi}} \triangleq \dot{\psi} - \dot{\psi}_d$ grows with $\dot{\psi}$ which leads to a necessary condition for a well defined yaw rate error. The yaw rate error dynamics are given by
\begin{align} \label{COL-eq:27}
\begin{split}
\dot{\tilde{\psi}} = &~r \left[ 1 + \frac{ X(u_r) u_{rd} } { u^2_{rd} + v^2_r } - \frac{ \Delta }{\Delta^2 + \left( y_{b/p} + g \right)^2 } \frac{\partial g}{\partial a} \left( 2 v_r X(u_r) \right) \right]  \\ &- \kappa(\theta) \dot{\theta} + \frac{ Y(u_r) v_r  u_{rd} - \dot{u}_{rd} v_r } { u^2_{rd} + v^2_r }  \\ &+ \frac{ \Delta }{\Delta^2 + \left( y_{b/p} + g \right)^2 }  \dot{\hat{V}}_N \frac{ b + \sqrt{b^2-ac} }{-a}  \\ &+ \frac{ \Delta }{ \Delta^2 + \left( y_{b/p} + g \right)^2 } \frac{\partial g}{\partial a} \left( 2 \hat{V}_N \dot{\hat{V}}_N - 2 u_{rd} \dot{u}_{rd} - 2 v_r Y(u_r) v_r \right) \\ &+ \frac{ \Delta }{\Delta^2 + \left( y_{b/p} + g \right)^2 } \frac{\partial g}{\partial b} \left( 2 \dot{\hat{V}}_N y_{b/p} \right)  \\ &+ \left[ 1 + \frac{\partial g}{\partial c} 2 y_{b/p} +  \frac{\partial g}{\partial b} \left( 2 \hat{V}_N \right) \right] \frac{ \Delta \dot{y}_{b/p}}{\Delta^2 + \left( y_{b/p} + g \right)^2 }
\end{split}
\end{align}
which leads to the following necessary condition for a well defined yaw rate, i.e. existence of the yaw controller,
\begin{condition} \label{COL-cond:Cr}
To have a well defined yaw controller it should hold that
\begin{equation}
C_r \triangleq 1 + \frac{ X(u_r) u_{rd} } { u^2_{rd} + v^2_r } - \frac{\partial g}{\partial a} \frac{ 2 v_r X(u_r)\Delta }{\Delta^2 + \left( y_{b/p} + g \right)^2 }  \neq 0.
\end{equation}
for all time after entering the tube.
\end{condition}

\begin{remark}
The condition above can be verified for any positive velocity, for the vehicles considered in this thesis. Note that for most vessels this condition is verifiable since standard ship design practices will result in similar properties of the function $X(u_r)$. Besides having a lower bound greater then zero $C_r$ is also upper-bounded since the term between brackets can be verified to be bounded in its arguments.
\end{remark}

Since $\dot{\psi}_d$ depends on the unknown signal $\tilde{V}_N$ we cannot take $\dot{\psi}_d = r_d$. To define an expression for $r_d$ without requiring the knowledge of $\tilde{V}_N$ we use \eqref{COL-eq:27} to define 
\begin{align} \label{COL-eq:rd}
\begin{split}
r_d \triangleq & -\frac{1}{C_r} \left[ \kappa(\theta) \left( \frac{u_t \cos(\psi + \beta - \gamma_p(\theta)) + k_\delta x_{b/p} + \hat{V}_{T}}{ 1 - \kappa(\theta) y_{b/p} } \right) \right. \\ &+ \left.  \frac{ Y(u_r) v_r  u_{rd} - \dot{u}_{rd} v_r } { u^2_{rd} + v^2_r } + \frac{ \Delta }{\Delta^2 + \left( y_{b/p} + g \right)^2 } \left[  \dot{\hat{V}}_{N} \frac{ b + \sqrt{b^2-ac} }{-a} \right. \right. \\ &+ \left. \frac{\partial g}{\partial a} \left( 2 \hat{V}_{N} \dot{\hat{V}}_{N} - 2 u_{rd} \dot{u}_{rd} - 2 v_r Y(u_r) v_r \right) + \frac{\partial g}{\partial b} \left( 2 \dot{\hat{V}}_{N} y_{b/p} \right) \right. \\ &+ \left. \left. \left[ 1 + \frac{\partial g}{\partial c} 2 y_{b/p}  + \frac{\partial g}{\partial b} 2 \hat{V}_{N} \right]   \left(  \frac{- u_{td}  y_{b/p} } { \sqrt{\Delta^2 + (y_{b/p} + g)^2 } } + G_1(\cdot)  \right) \right]\right]	
\end{split}
\end{align}
which results in the following yaw angle error dynamics
\begin{align} \label{COL-eq:dpsitilde}
\dot{\tilde{\psi}} = &~C_r \tilde{r} + \left[ 1 + \frac{\partial g}{\partial c} 2 y_{b/p} +  \frac{\partial g}{\partial b} \left( 2 \hat{V}_{N} \right) \right] \frac{ \Delta \tilde{V}_{N} }{\Delta^2 + \left( y_{b/p} + g \right)^2 } 
\end{align} 
where $\tilde{r} \triangleq r - r_d$ is the yaw rate error. From \eqref{COL-eq:dpsitilde} it can be seen that choosing $r_d$ as in \eqref{COL-eq:rd} results in yaw angle error dynamics that have a term dependent on the yaw rate error $\tilde{r}$ and a perturbing term that vanishes when the estimation error $\tilde{V}_N$ goes to zero. 

To add acceleration feedforward to the yaw rate controller, the derivative of $r_d$ needs to be calculated. However, when we analyse the dependencies of $r_d$ we obtain
\begin{align} \label{COL-eqeq:rd} 
r_d = & r_d( h, y_{b/p}, x_{b/p}, \tilde{\psi}, \tilde{x}, \tilde{y} ),
\end{align}
where $h = [\theta, v_r, u_r, u_{rd}, \dot{u}_{rd}, \hat{V}_{T}, \hat{V}_{N}]^T$ is introduced for the sake of brevity and represents all the variables whose derivatives do not contain $\tilde{V}_N$ or $\tilde{V}_T$. Consequently, the acceleration feedforward cannot be taken as $\dot{r}_d$ since using \eqref{COL-eqeq:rd}, \eqref{COL-eq:ddeltax}, and \eqref{COL-eq:dye} it is straightforward to verify this signal contains the unknowns $\tilde{V}_T$ and $\tilde{V}_N$. Therefore we define the yaw rate controller in terms of only known signals as:
\begin{align} \label{COL-eq:tau_r}
\begin{split}
\tau_r = & - F(u_r, v_r, r) + \frac{\partial r_d}{\partial h^T} \dot{h} + \frac{\partial r_d}{\partial y_{b/p}} \left( - u_{td}  \frac{ y_{b/p} } { \sqrt{\Delta^2 + (y_{b/p} + g)^2 } } + G_1(\cdot) \right) \\ &  + \frac{\partial r_d}{\partial x_{b/p}} \left( -k_\delta x_{b/p} \right) + 
\frac{\partial r_d}{ \partial \tilde{\psi}} C_r \tilde{r}  - \frac{\partial r_d}{ \partial \tilde{x}} k_x \tilde{x} - \frac{\partial r_d}{ \partial \tilde{y}} k_y \tilde{y} - k_1 \tilde{r} - k_2 \tilde{\psi}
\end{split}
\end{align}
Using \eqref{COL-eq:tau_r} in \eqref{COL-eq:rdot} we then obtain the yaw rate error dynamics
\begin{align} \label{COL-eq:32}
\begin{split}
\dot{\tilde{r}} = & - k_1 \tilde{r} - k_2C_r \tilde{\psi}  - \frac{\partial r_d}{ \partial \tilde{\psi}} \left[ 1 + \frac{\partial g}{\partial c} 2 y_{b/p}  +  \frac{\partial g}{\partial b} \left( 2 \hat{V}_{N} \right)  \right] \frac{ \Delta \tilde{V}_{N} }{\Delta^2 + \left( y_{b/p} + g \right)^2 } \\ & - \frac{\partial r_d}{\partial y_{b/p}} \tilde{V}_{N}  - \frac{\partial r_d}{\partial x_{b/p}} \tilde{V}_{T} + \frac{\partial r_d}{\partial \tilde{x}} \tilde{V}_x + \frac{\partial r_d}{ \partial \tilde{y}} \tilde{V}_y
\end{split}
\end{align}
which has a term depending on the yaw angle error, a term depending on the yaw rate error, and perturbing terms depending on the unknown ocean current estimation error.

\begin{remark}
It is straightforward to verify that all the terms in \eqref{COL-eq:dpsid} are smooth fractionals that are bounded with respect to $(y_{b/p}$, $x_{b/p}$, $\tilde{x}$, $\tilde{y}$, $\tilde{\psi})$ or are periodic functions with linear arguments and consequently the partial derivatives \eqref{COL-eq:tau_r} and \eqref{COL-eq:32} are all bounded. This is something that is used when showing closed-loop stability in the next section.
\end{remark}

%

\section{Closed-Loop Analysis} \label{COL-sec:clsys}

In this section we analyse the closed-loop system of the model \eqref{COL-eq:dynsys} with controllers \eqref{COL-eq:tauu} and \eqref{COL-eq:tau_r} and observer \eqref{COL-eq:obs} when the frame propagates with \eqref{COL-eq:theta2} along the path $P$. To show that path following is achieved we have to show that the following error dynamics converge to zero
\begin{subequations}  \label{COL-eq:clerr}
\begin{align}
\dot{x}_{b/p} = & -k_\delta x_{b/p}+ \tilde{V}_T \\
\dot{y}_{b/p} = & - u_{td}  \frac{ y_{b/p} } { \sqrt{\Delta^2 + (y_{b/p} + g)^2 } } + G_1(\cdot) + \tilde{V}_{N} \\
\dot{\tilde{\psi}} = &~C_r \tilde{r} + \left[ 1 + \frac{\partial g}{\partial c} 2 y_{b/p} + \frac{\partial g}{\partial b} \left( 2 \hat{V}_{N} \right)  \right] \frac{ \Delta \tilde{V}_{N} }{\Delta^2 + \left( y_{b/p} + g \right)^2 } \\
\begin{split}
\dot{\tilde{r}} = & - k_1 \tilde{r} - k_2C_r \tilde{\psi} - \frac{\partial r_d}{\partial y_{b/p}} \tilde{V}_{N}  - \frac{\partial r_d}{\partial x_{b/p}} \tilde{V}_{T} + \frac{\partial r_d}{\partial \tilde{x}} \tilde{V}_x + \frac{\partial r_d}{ \partial \tilde{y}} \tilde{V}_y \\ & - \frac{\partial r_d}{ \partial \tilde{\psi}} \left[ 1 + \frac{\partial g}{\partial c} 2 y_{b/p} + \frac{\partial g}{\partial b} \left( 2 \hat{V}_{N} \right)  \right] \frac{ \Delta \tilde{V}_{N} }{\Delta^2 + \left( y_{b/p} + g \right)^2 }  \end{split} \\
\dot{\tilde{u}} = & -\left(k_u + \frac{d_{11}}{m_{11}}\right) \tilde{u}
\end{align}
\end{subequations}
The system \eqref{COL-eq:clerr} has the following perturbed form:
\begin{align} \label{COL-eq:clerrpt}
\dot {\tilde{X}} \triangleq & \begin{bmatrix} \dot{x}_{b/p} \\ \dot{y}_{b/p} \\  \dot{\tilde{\psi}} \\  \dot{\tilde{r}} \\ \dot{\tilde{u}} \end{bmatrix}  = \begin{bmatrix} -k_{\delta}x_{b/p} \\  - u_{td}  \frac{ y_{b/p} } { \sqrt{\Delta^2 + (y_{b/p} + g)^2 } } + G_1(\cdot) \\  C_r \tilde{r} \\ - k_1 \tilde{r} - k_2 C_r\tilde{\psi} \\ -k_3 \tilde{u}   \end{bmatrix} + \nonumber \\ & \begin{bmatrix} \tilde{V}_T \\ \tilde{V}_{N} \\ \left[ 1 + \frac{\partial g}{\partial c} 2 y_{b/p} + \frac{\partial g}{\partial b} \left( 2 \hat{V}_{N} \right)  \right] \frac{ \Delta \tilde{V}_{N} }{\Delta^2 + \left( y_{b/p} + g \right)^2 } \\ - \frac{\partial r_d}{\partial \bsym{p}_{b/p}}\begin{bmatrix}\tilde{V}_{T}\\\tilde{V}_{N}\end{bmatrix} - \frac{\partial r_d}{ \partial \tilde{\psi}} \left[ 1 + \frac{\partial g}{\partial c} 2 y_{b/p} + \frac{\partial g}{\partial b} 2 \hat{V}_{N} \right] \frac{ \Delta \tilde{V}_{N} }{\Delta^2 + \left( y_{b/p} + g \right)^2 } - \frac{\partial r_d}{ \partial \tilde{\bsym{p}_{b/p}}} \tilde{\bsym{V}}_c \\ 0  \end{bmatrix}
\end{align} 
where $\bsym{p}_{b/p} \triangleq [x_{b/p},y_{b/p}]^T$ and all the perturbing terms disappear as the current estimates converge to zero. In particular, we cannot apply our desired control action whilst the current estimates have not converged yet, since the current cannot be compensated for until it is estimated correctly.

The full closed-loop system of the model \eqref{COL-eq:dynsys} with controllers \eqref{COL-eq:tauu} and \eqref{COL-eq:tau_r} and observer \eqref{COL-eq:obs} is given by 
\begin{subequations} \label{COL-eq:clfull}
\begin{align} 
\dot {\tilde{X}}_1 &\triangleq \begin{bmatrix} \dot{y}_{b/p} \\  \dot{\tilde{\psi}} \\  \dot{\tilde{r}} \end{bmatrix}  = \begin{bmatrix}  - u_{td}  \frac{ y_{b/p} } { \sqrt{\Delta^2 + (y_{b/p} + g)^2 } } + G_1(\cdot) \\  C_r \tilde{r} \\ - k_1 \tilde{r} - k_2C_r \tilde{\psi} \end{bmatrix} + \nonumber \\ & \begin{bmatrix} \tilde{V}_{N} \\ \left[ 1 + \frac{\partial g}{\partial c} 2 y_{b/p} + \frac{\partial g}{\partial b} \left( 2 \hat{V}_{N} \right)  \right] \frac{ \Delta \tilde{V}_{N} }{\Delta^2 + \left( y_{b/p} + g \right)^2 } \\ - \frac{\partial r_d}{\partial \bsym{p}_{b/p}}\begin{bmatrix}\tilde{V}_{T}\\\tilde{V}_{N}\end{bmatrix} - \frac{\partial r_d}{ \partial \tilde{\psi}} \left[ 1 + \frac{\partial g}{\partial c} 2 y_{b/p} + \frac{\partial g}{\partial b} 2 \hat{V}_{N} \right] \frac{ \Delta \tilde{V}_{N} }{\Delta^2 + \left( y_{b/p} + g \right)^2 } - \frac{\partial r_d}{ \partial \tilde{\bsym{p}}} \tilde{\bsym{V}}_c \end{bmatrix} \label{COL-eq:35a}  \\
\dot {\tilde{X}}_2 &\triangleq \begin{bmatrix} \dot{x}_{b/p} \\  \dot{\tilde{x}} \\ \dot{\tilde{y}} \\ \dot{ \tilde{V}}_x \\ \dot{ \tilde{V}}_y \\ \dot{\tilde{u}} \end{bmatrix}  =  \begin{bmatrix} - k_\delta x_{b/p}   + \tilde{V}_{T}  \\  -k_{x} \tilde{x} - \tilde{V}_{x} \\  -k_{y} \tilde{y} - \tilde{V}_{y} \\  - k_{x1} \tilde{x} \\ - k_{y1} \tilde{y} \\ -k_u \tilde{u} \end{bmatrix}  \label{COL-eq:35b} \\
\dot{v}_r &=   X ( u_{rd} + \tilde{u} ) r_d ( h, y_{b/p}, x_{b/p}, \tilde{\psi}, \tilde{x}, \tilde{y}) + X( u_{rd} + \tilde{u} ) \tilde{r} + Y( u_{rd} + \tilde{u} ) v_r  \label{COL-eq:35c}
\end{align} 
\end{subequations}

Before starting with the stability analysis of \eqref{COL-eq:clfull}, we first establish GES of \eqref{COL-eq:35b} by using the following lemma.  

\begin{lemma} \label{COL-lemx}
The system \eqref{COL-eq:35b} is GES.
\end{lemma}
\begin{proof}
Note that \eqref{COL-eq:35b} is a cascaded system of the form
\begin{subequations} \label{COL-eq:lemxcasc}
\begin{align}
\dot{x}_{b/p} &= - k_\delta x_{b/p}   + \tilde{V}_{T}, \label{COL-eq:lemxcasca} \\ 
\begin{bmatrix}  \dot{\tilde{x}} \\ \dot{\tilde{y}} \\ \dot{ \tilde{V}}_x \\ \dot{ \tilde{V}}_y \\ \dot{\tilde{u}} \end{bmatrix} & =  \begin{bmatrix}   -k_{x} \tilde{x} - \tilde{V}_{x} \\  -k_{y} \tilde{y} - \tilde{V}_{y} \\  - k_{x1} \tilde{x} \\ - k_{y1} \tilde{y} \\ -k_u \tilde{u} \end{bmatrix}. \label{COL-eq:lemxcascb}
\end{align}
\end{subequations}
The nominal dynamics of \eqref{COL-eq:lemxcasc} are given by $\dot{x}_{b/p} = - k_\delta x_{b/p}$ from \eqref{COL-eq:lemxcasca}, which is a stable linear system and thus GES. The perturbing dynamics are given by \eqref{COL-eq:lemxcascb} and where shown to be GES in Section 4 of the paper. The interconnection term is the term $\tilde{V}_T$ from \eqref{COL-eq:lemxcasca}. The growth of the interconnection term can be bounded by $\Vert \tilde{V}_T \Vert \leq \Vert [\tilde{V}_x,\tilde{V}_y]^T\Vert$, which satisfies the condition for the interconnection term from \citet[Theorem~2]{panteley1998global}. Note that it is trivial to shown the nominal dynamics admit the quadratic Lyapunov function $V_{x_{b/p}} = 1/2x^2_{b/p}$. Consequently, all the conditions of \citet[Theorem~2]{panteley1998global} and \citet[Proposition~2.3]{loria2004cascaded} are satisfied. Therefore, the cascaded system \eqref{COL-eq:lemxcasc} is GES, which implies that \eqref{COL-eq:35b} is GES.
\end{proof}

Note that although we show that the system \eqref{COL-eq:35b} is GES, the dynamics of $x_{b/p}$ are only defined in the tube to avoid the singularity in the parametrisation. Hence, the stability result is only valid in the tube.

The first step in the stability analysis of \eqref{COL-eq:clfull} is to assure that the closed-loop system is forward complete and that the sway velocity $v_r$ remains bounded. Therefore, under the assumption that Condition \ref{COL-cond:curv}-\ref{COL-cond:Cr} are satisfied, i.e. $ 1-\kappa(\theta)y_{b/p} \neq 0 $ and $ C_r \neq 0 $, we take the following three steps:
\begin{enumerate}
\item First, we prove that the trajectories of the closed-loop system are forward complete. 
\item Then, we derive a necessary condition such that $v_r$ is locally bounded with respect to $( \tilde{X}_1, \tilde{X}_2 )$.
\item Finally, we establish that for a sufficiently big value of $\Delta$, $v_r$ is locally bounded only with respect to $\tilde{X}_2$. 
\end{enumerate}

The above three steps are taken by formulation and proving three lemmas. For the sake of brevity in the main body of this paper the proofs of the following lemmas are replaced by a sketch of each proof in the main body. The full proofs are reported to the Appendix.

\begin{lemma} [Forward completeness] \label{COL-lem1}
The trajectories of the global closed-loop system \eqref{COL-eq:clfull} are forward complete.
\end{lemma}
The proof of this lemma is given in the Appendix
. The general idea is as follows. Forward completeness for \eqref{COL-eq:35b} is evident since this part of the closed-loop system consists of GES error dynamics. Using the forward completeness and in fact boundedness of \eqref{COL-eq:35b} we can show forward completeness of \eqref{COL-eq:35c}, $\dot{\tilde{\psi}}$, and $\dot{\tilde{r}}$. Hence, forward completeness of \eqref{COL-eq:clfull} depends on forward completeness of $\dot{y}_{b/p}$. To show forward completeness of $\dot{y}_{b/p}$, we consider the $y_{b/p}$ dynamics with $\tilde{X}_2$, $\tilde{\psi}$, $\tilde{r}$, and $v_r$ as input, which allows us to claim forward completeness of $\dot{y}_{b/p}$. Consequently, all the states of the closed-loop system are forward complete and hence the closed-loop system \eqref{COL-eq:clfull} is forward complete

\begin{lemma} [Boundedness near $(\tilde{X}_1, \tilde{X}_2)=0$] \label{COL-lem2}
The system \eqref{COL-eq:35c} is bounded near $(\tilde{X}_1, \tilde{X}_2)=0$ if and only if the curvature of $P$ satisfies the following condition:
\begin{equation} \label{COL-eq:curvlem2}
\kappa_{\max} \triangleq \max_{\theta \in P}\left| \kappa(\theta) \right| < \frac{ Y_{\min} }{ X_{\max} }.
\end{equation}
\end{lemma}

The proof of this lemma is given in the Appendix
. A sketch of the proof is as follows. The sway velocity dynamics \eqref{COL-eq:35c} are analysed using a quadratic Lyapunov function $V=1/2v^2_r$. It can be shown that the derivative of this Lyapunov function satisfies the conditions for boundedness when the solutions are on or close to the manifold where $(\tilde{X}_1, \tilde{X}_2)=0$. Consequently, \eqref{COL-eq:35c} satisfies the conditions of boundedness near $(\tilde{X}_1,\tilde{X}_2)=0$ as long as \eqref{COL-eq:curvlem2} is satisfied.

In Lemma \ref{COL-lem2} we show boundedness of $v_r$ for small values of $(\tilde{X}_1, \tilde{X}_2)$ to derive the bound on the curvature. However, locality with respect to $\tilde{X}_1$, i.e. the path-following errors and yaw angle and yaw rate errors, is not desired and in the next lemma boundedness independent of $\tilde{X}_1$ is shown under an extra condition on the look-ahead distance $\Delta$.

\begin{lemma} [Boundedness near $\tilde{X}_2=0$] \label{COL-lem3}
If the following additional assumption is satisfied:
\begin{equation}
\exists ~\sigma > 0 ~ \mathrm{s.t.} \quad 1 - \kappa(\theta) y_{b/p} \geq \sigma > 0 \quad \wedge \quad \left[ Y_{\min} - X_{\max} \kappa_{\max} \frac{ 1 }{ \sigma } \right] > 0 \label{COL-eq:sigma}
\end{equation}
the system \eqref{COL-eq:35c} is bounded only near $ \tilde{X}_2 = 0$ if we have
\begin{gather} 
\Delta > \frac{ 4 X_{\max} }{ \left[ Y_{\min} - X_{\max} \kappa_{\max} \frac{ 1 }{ \sigma } \right] } \label{COL-eq:Delta} \\
\kappa_{\max} < \sigma \frac{ Y_{\min}}{X_{\max} } \label{COL-eq:curvlem3}
\end{gather}
\end{lemma}

\begin{remark} \label{COL-rem:sig}
The size of $\sigma$ can be calculated by using the following tuning procedure. 
\begin{enumerate}
\item Start by calculating the absolute bound on the curvature from Lemma \ref{COL-lem2}. This is a bound that is necessary for feasibility of the trajectories. 
\item Now choose a positive $\Delta$ and using the maximum curvature of the path, solve \eqref{COL-eq:Delta} to obtain a possible value for $\sigma$.
\item Using the value for $\sigma$ obtained in the previous step and the maximum value of the curvature we can use the inequality $1-\kappa(\theta)y_{b/p}\geq \sigma$ from \eqref{COL-eq:sigma} to calculate the size of the tube as
\begin{equation}
y^{\mrm{tube}}_{b/p} = \frac{1-\sigma}{\kappa_{\max}}.
\end{equation}
\end{enumerate}
If initial conditions are within the tube $y^{\mrm{tube}}_{b/p}$, and are chosen such that the transient caused by the unknown current does not force the vessel out of the tube. Then the sway velocity is bounded for all time. Note that the choice of $\Delta$ in step two given above determines how large the tube will be. More specifically, a larger choice for $\Delta$ will result in a smaller value for $\sigma$ which will lead to a larger tube in step three. However, due to the nature of the guidance a larger $\Delta$ will mean slower steering and consequently slower convergence to the path.
\end{remark}

The proof of Lemma \ref{COL-lem3} is given in the Appendix
, the general idea is given as follows. The proof follows along the same lines of that of Lemma \ref{COL-lem2} but solutions are considered close to the manifold $\tilde{X}_2 = 0$ rather than $(\tilde{X}_1, \tilde{X}_2)=0$. It is shown that boundedness can still be shown if \eqref{COL-eq:Delta} is satisfied additionally to the conditions of Lemma \ref{COL-lem2}. 

\begin{theorem}
Consider a $\theta$-parametrised path denoted by $ P(\theta) \triangleq (x_p(\theta), y_p(\theta))$. Then under Conditions \ref{COL-cond:curv}-\ref{COL-cond:Cr} and the conditions of Lemma \ref{COL-lem1}-\ref{COL-lem3}, the system \eqref{COL-eq:dynsys} with control laws \eqref{COL-eq:tauu} and \eqref{COL-eq:tau_r} and observer \eqref{COL-eq:obs} follows the path $P$, while maintaining $v_r$, $\tau_r$ and $\tau_u$ bounded. In particular, the origin of the system \eqref{COL-eq:35a}-\eqref{COL-eq:35b} is exponentially stable in the tube. 
\end{theorem}
\begin{proof}
From the fact that the origin of \eqref{COL-eq:35b} is GES, the fact that the closed-loop system \eqref{COL-eq:clfull} is forward complete according to Lemma \ref{COL-lem1}, and the fact that solutions of \eqref{COL-eq:35c} are locally bounded near $\tilde{X}_2 = 0$ according to Lemma \ref{COL-lem3}, we can conclude that there is a finite time $T > t$ after which solutions of \eqref{COL-eq:35b} will be sufficiently close to $\tilde{X}_2 = 0$ to guarantee boundedness of $v_r$. Having established that $v_r$ is bounded we first analyse the cascade 
\begin{subequations} \label{COL-eq:clstab}
\begin{align} 
\begin{bmatrix}\begin{smallmatrix} \dot{\tilde{\psi}} \\  \dot{\tilde{r}} 
\end{smallmatrix}\end{bmatrix}  &=  
\begin{bmatrix}\begin{smallmatrix} C_r \tilde{r} \\ - k_1 \tilde{r} - k_2 C_r \tilde{\psi} 
\end{smallmatrix}\end{bmatrix} + \notag \\
&\quad \begin{bmatrix}\begin{smallmatrix} G_2(\cdot)  \\ - \tfrac{\partial r_d}{ \partial \tilde{\psi}} G_2(\cdot) - \tfrac{\partial r_d}{\partial \bsym{p}_{b/p}} [\tilde{V}_{T},\tilde{V}_N]^T + \tfrac{\partial r_d}{\partial [\tilde{x},\tilde{y}]^T} \tilde{\bsym{V}}_c  \end{smallmatrix}\end{bmatrix} \label{COL-eq:clstaba}\\
\begin{bmatrix}\begin{smallmatrix}
\dot{x}_{b/p} \\  \dot{\tilde{x}} \\ \dot{\tilde{y}} \\ \dot{ \tilde{V}}_x \\ \dot{ \tilde{V}}_y \\ \dot{\tilde{u}} 
\end{smallmatrix}\end{bmatrix}  &=  
\begin{bmatrix}\begin{smallmatrix} - k_\delta x_{b/p}   + \tilde{V}_{T}  \\  -k_{x} \tilde{x} - \tilde{V}_{x} \\  -k_{y} \tilde{y} - \tilde{V}_{y} \\  - k_{x1} \tilde{x} \\ - k_{y1} \tilde{y} \\ -k_u \tilde{u} 
\end{smallmatrix}\end{bmatrix}  \label{COL-eq:clstabb}
\end{align} 
\end{subequations}
The perturbing system \eqref{COL-eq:clstabb} is GES as shown in Lemma \ref{COL-lemx}. The interconnection term, i.e. the second and third term in \eqref{COL-eq:clstaba}, satisfies the linear growth criteria from \cite[Theorem 2]{panteley1998exponential}. More specifically, it has an upperbound that does not grow with $\tilde{\psi}$ and $\tilde{r}$ since all the partial derivatives of $r_d$ and $g$ can be bounded by constants. The nominal dynamics, i.e. the first matrix in \eqref{COL-eq:clstaba}, can be analysed with the following quadratic Lyapunov function
$
V_{(\tilde{r},\tilde{\psi})} = \tfrac{1}{2}\tilde{r}^2 + \tfrac{1}{2}k_2\tilde{\psi}^2
$,
whose derivative along the solutions of the nominal system is given by
\begin{equation} \label{COL-eq:dV1thm}
\dot{V}_{(\tilde{r},\tilde{\psi})} = - k_1\tilde{r}^2 -k_2C_r\tilde{\psi}\tilde{r} + k_2C_r\tilde{r}\tilde{\psi} = -k_2 \tilde{r}^2 \leq 0
\end{equation}
which implies that $\tilde{r}$ and $\tilde{\psi}$ are bounded. The derivative of \eqref{COL-eq:dV1thm} is given by
\begin{equation}
\ddot{V}_{(\tilde{r},\tilde{\psi})} = -2k^2_1 \tilde{r}^2 - 2k_1k_2C_r \tilde{\psi}\tilde{r}
\end{equation}
which is bounded since $\tilde{r}$ and $\tilde{\psi}$ are bounded. This implies that \eqref{COL-eq:dV1thm} is a uniformly continuous function. Consequently, from Barbalat's lemma (\citet[Lemma~8.2]{khalil2002nonlinear}) we have that
\begin{equation}
\lim_{t\rightarrow\infty} \dot{V}_{(\tilde{r},\tilde{\psi})} = \lim_{t\rightarrow\infty} -k_1\tilde{r}^2 = 0~\Rightarrow~\lim_{t \rightarrow \infty} \tilde{r} = 0.
\end{equation}
Since $C_r$ is persistently exciting, which follows from the fact that $C_r$ is upper bounded and lower bounded by positive constants, it follows from the expression of the nominal dynamics that
\begin{equation}
\lim_{t \rightarrow \infty} \tilde{r} = 0~\Rightarrow~\lim_{t \rightarrow \infty} \tilde{\psi} = 0.
\end{equation}
This implies that the system is globally asymptotically stable. Consequently, from the above it follows that the cascade \eqref{COL-eq:clstab} is GES \cite[Theorem 2]{panteley1998exponential}.

We now consider the following dynamics
\begin{equation} \label{COL-eq:clstab2}
\dot{y}_{b/p} =  - u_{td}   \tfrac{ y_{b/p} }{ \sqrt{\Delta^2 + (y_{b/p} + g)^2} } + \tilde{V}_{N} + G_1(\cdot).
\end{equation}
Note that we can view the systems \eqref{COL-eq:clstab} and \eqref{COL-eq:clstab2} as a cascaded system where the nominal dynamics are formed by the first term of \eqref{COL-eq:clstab2}, the interconnection term is given by the second and third terms of \eqref{COL-eq:clstab2}, and the perturbing dynamics are given by \eqref{COL-eq:clstab}. As we have just shown, the perturbing dynamics are GES. Using the bound on $G_1(\cdot)$ from \eqref{COL-eq:Gbnd} it is straightforward to verify that the interconnection term satisfies the conditions of \cite[Theorem 2]{panteley1998exponential}. We now consider the following Lyapunov function for the nominal system
$
V_{y_{b/p}} = 1/2y^2_{b/p},
$
whose derivative along the solutions of the nominal system is given by
\begin{equation}
\dot{V}_{y_{b/p}} = - u_{td}   \tfrac{ y^2_{b/p} }{ \sqrt{\Delta^2 + (y_{b/p} + g)^2} } \leq 0,
\end{equation}
which implies that the nominal system is GAS. Moreover, since it is straightforward to verify that $\dot{V}_{y_{b/p}} \leq \alpha V_{y_{b/p}}$ for some constant $\alpha$ dependent on the initial conditions, it follows from the comparison lemma (\citet[Lemma 3.4]{khalil2002nonlinear}) that the nominal dynamics are also LES. Consequently, the cascaded system satisfies the conditions of \citet[Theorem~2]{panteley1998global} and \citet[Lemma~8]{panteley1998exponential}, and therefore the cascaded system is GAS and LES. This implies that the origin of the error dynamics, i.e. $(\tilde{X}_1,\tilde{X}_2)=(0,0)$, is globally asymptotically stable and locally exponentially stable. However, since the parametrisation is only valid locally we can only claim exponential stability in the tube.
\end{proof} 

\section{Case Study} \label{COL-sec:case}
This section presents a case study to verify the theoretical results presented in this paper. The case study under consideration is following of a circular path using the model of an underactuated surface vessel from \citet{fredriksen2004global}. The ocean current components are given by $V_x = -1~\munit{m/s}$ and $V_y = 1.2~\munit{m/s}$ and consequently $V_{\max} \approx 1.562~\munit{m/s}$. The desired relative surge velocity is chosen to be constant and set to $u_{rd}= 5~\munit{m/s}$ such that Assumption \ref{COL-assum:vel} is verified. Using the ship's model parameters from \citet{fredriksen2004global} and the expressions \eqref{COL-eq:Xur} and \eqref{COL-eq:Yur} it is straightforward to see that the curvature bound from Lemma \ref{COL-lem2} is given by $\kappa_{\max} < (Y_{\min})/(X_{\max}) \approx 0.1333$. The observer is initialised as suggested in Subsection \ref{COL-subsec:obs} and the observer gains are selected as $k_{x_1} = k_{y_1} = 1$ and $k_{x_1} = k_{y_1} = 0.1$. The controller gains are selected as $k_{u_r} = 0.1$ for the surge velocity controller and $k_{1} = 1000$ and $k_2 = 400$ for the yaw rate controller.

In this case study the vessel is required to follow a circle with a radius of $400~\munit{m}$ that is centred around the origin. Consequently, the curvature of the path is given by $\kappa_p = 1/400 = 0.0025$. To choose the parameters of the guidance law we will now follow the tuning procedure lined out in Remark \ref{COL-rem:sig}. In the first step we verify that the feasibility constraint on the curvature is satisfied for the path under consideration, which is clearly the case since $\kappa_p < (Y_{\min})/(X_{\max}) \approx 0.133$. In the second step we fix our $\Delta$ as $\Delta = 40~\munit{m}$, which results in $\sigma \approx 0.0268$. In the third step we use the value for $\sigma$ to calculate the size of the tube as $y^{\mrm{tube}}_{b/p} \approx 369.983~\munit{m}$. Note that this is only slightly smaller then the size of the tube where the parametrisation is valid, i.e. $400~\munit{m}$.  To stay within this tube we choose the initial conditions as
\begin{equation}
[u_r(t_0),v_r(t_0),r(t_0),x(t_0),y(t_0),\psi(t_0)]^T = [0,0,0,700,10,\pi/2]^T.
\end{equation}

The resulting trajectory for the vessel can be seen in Figure \ref{COL-fig:Loc_path}. The blue dashed line is the trajectory of the vessel and the red circle is the reference path. The yellow vessels represent the orientation of the vessel at certain instances. From the plot in Figure \ref{COL-fig:Loc_path} it can be seen that the vessel converges to the circle and starts to follow the path. Moreover, it can be seen from the yellow vessels that the orientation of the ship is not tangential to the circle which is necessary to compensate for the ocean current.

\begin{figure}[tbh]
\centering
\includegraphics[width=.75\columnwidth]{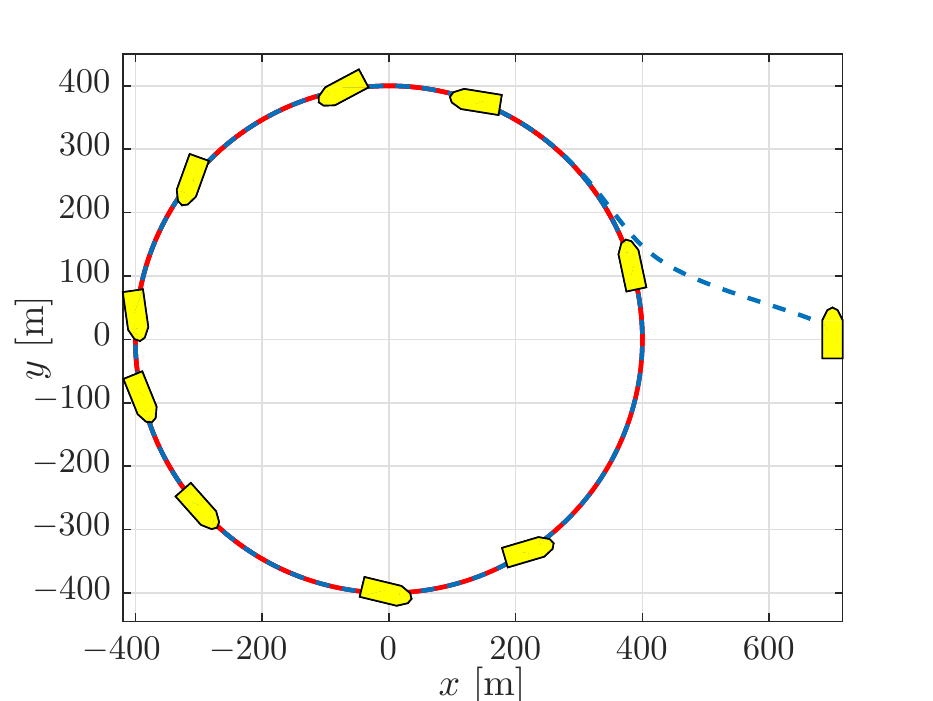}{}
\caption[Path of the vessel in the $x-y$-plane]{Path of the vessel in the $x-y$-plane. The dashed blue line is the trajectory of the path and the red line is the reference. The yellow ships denote the orientation of the vessel at certain times.}\label{COL-fig:Loc_path}
\end{figure}

The path-following errors can be seen in the top plot of Figure \ref{COL-fig:Loc_sub_cp} which confirm that the path-following errors converge to zero. A detail of the steady-state is given to show the reduction of the error. Moreover, note that because of the choice of parametrisation the error in tangential direction $x_{b/p}$ is zero throughout the motion except from a very small transient at the beginning caused by the transient of the observer. The estimates obtained from the ocean current observer can be seen in the second plot from the top in Figure \ref{COL-fig:Loc_sub_cp}. From this plot it can be seen that the estimates converge exponentially with no overshoot. This underlines the conservativeness of the bound from Assumption \ref{COL-assum:vel} that is required for the error bound for the observer as explained in Subsection \ref{COL-subsec:obs}. The third plot in Figure \ref{COL-fig:Loc_sub_cp} depicts the yaw rate and the sway velocity induced by the motion. It can be seen that these do not converge to zero but converge to a periodic motion. Note that for circular motion in the absence of current the yaw rate would converge to zero. However, when current is present the vessel needs to change its turning rate depending on if it goes with or against the current. The relative surge velocity is given in the fourth plot from the top in Figure \ref{COL-fig:Loc_sub_cp} and shows that the surge velocity converges exponentially to the desired value. This plot is especially interesting in combination with the plot of the magnitude of $C_r$ given at the bottom of Figure \ref{COL-fig:Loc_sub_cp}. From this plot it can clearly be seen that Condition \ref{COL-cond:Cr} is verified both in steady-state and during the transient of the velocity controller. 

\begin{figure}[p]
\centering
\includegraphics[width=0.85\textwidth]{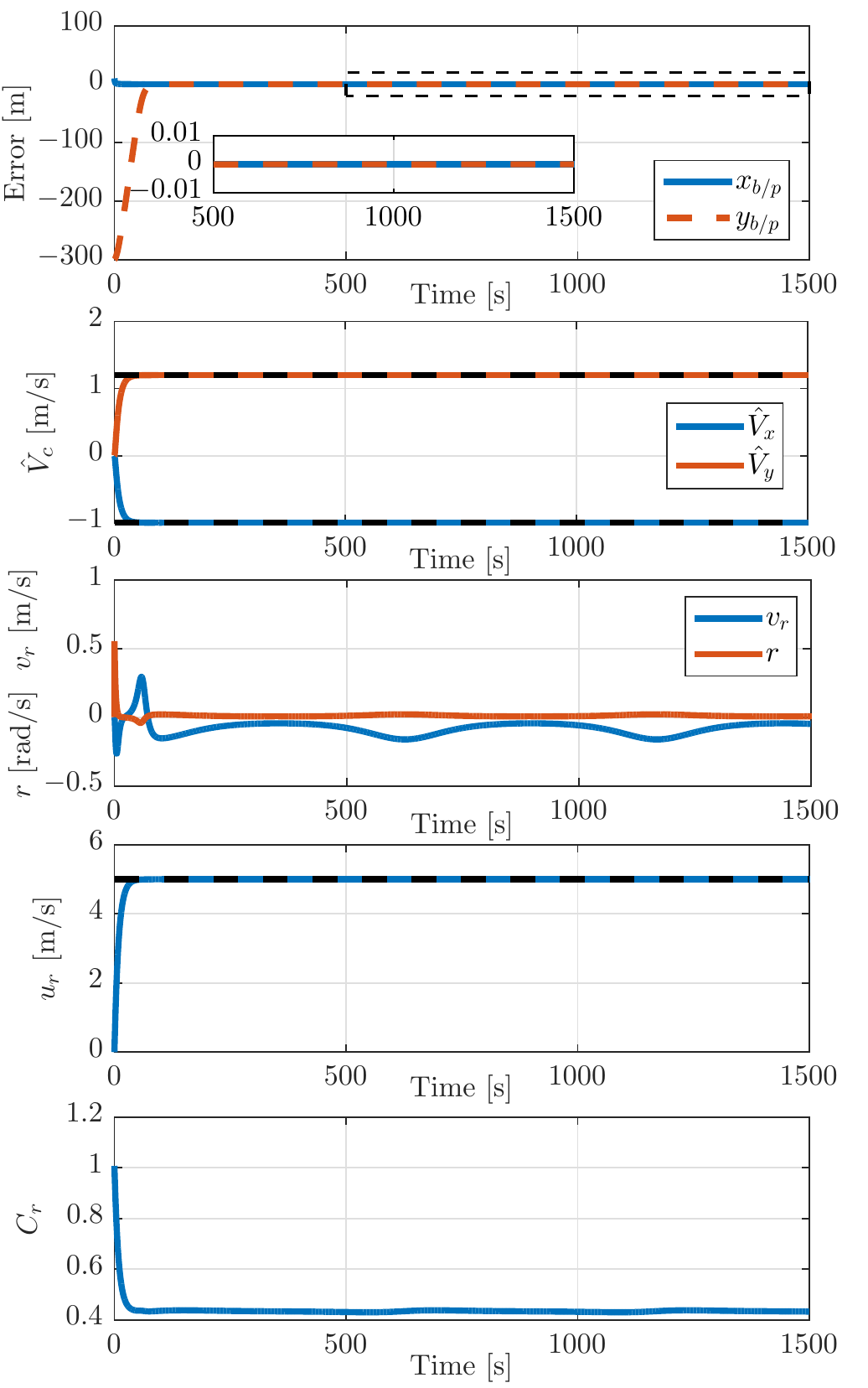}
\caption[Path following errros, current estimates, sway velocity, yaw rate, surge velocity , and size of $C_r$ over time.]{Path following errros plotted agains time (top), current estimates against time (second), sway velocity and yaw rate against time (third), surge velocity against time (fourth), and size of $C_r$ over time (bottom).}\label{COL-fig:Loc_sub_cp}
\end{figure}

\section{Conclusion} \label{COL-sec:cncl}
This paper considered curved-path following for underactuated marine vessels in the presence of constant ocean currents. In this approach the path is parametrised by a path variable with a update law that is designed to keep the vessel on the normal of a path-tangential reference frame. This assures the path-following error is defined as the shortest distance to the path. However, the disadvantage is that this type of update law has a singularity which only allows for local results. The vessel is steered using a line-of-sight guidance law, which to compensate for the unknown ocean currents is aided by an ocean current observer. The closed-loop system with the controllers and observer was analysed. This was done by first showing boundedness of the underactuated sway velocity dynamics under certain conditions. It was then shown that if these conditions are satisfied and the sway velocity is bounded the path-following errors are exponentially stable within the tube. Due to the singularity the feasibility of this problem depends on the initial conditions, the curvature of the path, and the magnitude of the ocean current. More specifically, the size of the tube in which the parametrisation is well defined was shown to be a function of the maximal curvature of the path. This implies that the combination of curvature and ocean current should be such that a suitable set of initial conditions exists for which the transient of the ocean current observer does not take the vessel out of the tube.

\appendix

\section*{Appendix}

\subsection*{Proof of Lemma \ref{COL-lem1}} \label{COL-app:lem1}

Consider the following part of the global closed-loop system:
\begin{subequations} \label{COL-eq:36}
\begin{align} 
\begin{split}
\begin{bmatrix}  \dot{\tilde{\psi}} \\  \dot{\tilde{r}} \end{bmatrix}  &=  \begin{bmatrix}  C_r \tilde{r} \\ - k_1 \tilde{r} - k_2C_r \tilde{\psi} \end{bmatrix} \\ &+ \underbrace{ \begin{bmatrix} \left[ 1 + \frac{\partial g}{\partial c} 2 y_{b/p} + \frac{\partial g}{\partial b} \left( 2 \hat{V}_{N} \right)  \right] \frac{ \Delta \tilde{V}_{N} }{\Delta^2 + \left( y_{b/p} + g \right)^2 } \\ - \frac{\partial r_d}{\partial \bsym{p}_{b/p}}\begin{bmatrix}\tilde{V}_{T}\\\tilde{V}_{N}\end{bmatrix} - \frac{\partial r_d}{ \partial \tilde{\psi}} \left[ 1 + \frac{\partial g}{\partial c} 2 y_{b/p} + \frac{\partial g}{\partial b} 2 \hat{V}_{N} \right] \frac{ \Delta \tilde{V}_{N} }{\Delta^2 + \left( y_{b/p} + g \right)^2 } - \frac{\partial r_d}{ \partial \tilde{\bsym{p}}_{b/p}} \tilde{\bsym{V}}_c  \end{bmatrix} }_{R(h, y_{b/p}, x_{b/p}, \tilde{\psi}, \tilde{x}, \tilde{y})} \label{COL-eq:36a} \end{split} \\
\dot{v}_r &=   X ( u_{rd} + \tilde{u} ) r_d ( h, y_{b/p}, x_{b/p}, \tilde{\psi}, \tilde{x}, \tilde{y}) + X( u_{rd} + \tilde{u} ) \tilde{r} + Y( u_{rd} + \tilde{u} ) v_r \label{COL-eq:36b}
\end{align}
\end{subequations}
From the boundedness of the vector $ [ \tilde{X}^T_2 , \kappa(\theta), u_{rd} , \dot{u}_{rd}, V_T, V_N ]^T $ we know that  $ \left\| [ \tilde{X}^T_2 , \kappa(\theta), u_{rd} , \dot{u}_{rd}, V_T, V_N ]^T \right\| \leq \beta_0 $, and from the expression for $r_d$ in the paper we can conclude the existence of positive functions $ a_{r_d}(\cdot) $, $ b_{r_d}(\cdot) $, $ a_{R}(\cdot) $, and $ b_R(\cdot) $ which are all continuous in their arguments and are such that such the following inequalities hold:
\begin{align} \label{COL-eq:born1}
 \left| r_d(\cdot) \right| \leq &~a_{r_d}(\Delta, \beta_0) \left| v_r \right|+ b_{r_d} (\Delta, \beta_0)	
\end{align}
and,
\begin{align} \label{COL-eq:born2}
\left\| R(\cdot) \right\|	\leq a_R( \Delta, \beta_0 ) \left| v_r \right| + b_R ( \Delta, \beta_0 )
\end{align}
Then taking the following Lyapunov function candidate:
\begin{align} \label{COL-eq:V}
V_1(\tilde{\psi}, \tilde{r}, v_r) = \frac{1}{2}	 \left( k_2\tilde{\psi}^2 + \tilde{r}^2 + v^2_r  \right)
\end{align}
whose time derivative along the solutions of \eqref{COL-eq:36} is 
\begin{align} \label{COL-eq:V1}
\begin{split}
\dot{V}_1(\cdot) = &~k_2C_r \tilde{r} \tilde{\psi} - k_1 \tilde{r}^2 - k_2C_r \tilde{r} \tilde{\psi} + [\tilde{\psi}~~\tilde{r}] R(\cdot) + Y(u_{rd}+ \tilde{u}) v^2_r \\ &+ X(u_{rd} + \tilde{u}) \tilde{r} v_r + X(u_{rd}+\tilde{u}) r_d(\cdot) v_r 
\end{split}
\end{align}
Using Young's inequality we note that
\begin{align} \label{COL-eq:dVfc}
\begin{split}
\dot{V}_1(\cdot) \leq &~k_1 \tilde{r}^2 + \tilde{\psi}^2 + \tilde{r}^2 + R^2(\cdot) + Y(u_{rd} + \tilde{u}) v^2_r \\ &+ \left| X( u_{rd} + \beta_0 ) \right|  \left( \tilde{r}^2 + v^2_r \right) + \left| X (u_{rd} + \beta_0) \right| \left( r^2_d(\cdot) + v^2_r \right) \\ \leq & \alpha V + \beta ,\; \alpha \geq 0, \;\beta \geq 0
\end{split}
\end{align}
Note that since the differential inequality \eqref{COL-eq:dVfc} is scaler we can invoke the comparison lemma (\citet[Lemma 3.4]{khalil2002nonlinear}). From the comparison lemma we know that the solutions of differential inequality \eqref{COL-eq:dVfc} are bounded by the solutions of the linear system:
\begin{equation}
\dot{x} = \alpha x + \beta 
\end{equation}
which has solutions
\begin{equation}
x(t) = \frac{\|x(t_0)\|\alpha + \beta}{\alpha}e^{\alpha(t-t_0)} - \frac{\beta}{\alpha}
\end{equation}
Hence, from the comparison lemma we have that 
\begin{equation}
V_1(\cdot) \leq \frac{\|V_1(t_0)\|\alpha + \beta}{\alpha}e^{\alpha(t-t_0)} - \frac{\beta}{\alpha}
\end{equation}
which shows the solutions of $V_1(\cdot)$ are defined up to $t_{\max} = \infty$ and consequently from \eqref{COL-eq:V} it follows that the solutions of $\tilde{\psi}$, $\tilde{r}$, and $v_r$ must be defined up to $t_{\max} = \infty$. Hence, the solutions of \eqref{COL-eq:36} satisfy the definition of forward completeness (\citet{angeli1999forward}) and we can conclude forward completeness of trajectories of \eqref{COL-eq:36}. 

The forward completeness of trajectories of the global closed-loop system now depends on forward completeness of $\dot{y}_{b/p}$ from \eqref{COL-eq:35a}. We can conclude forward completeness of $\dot{y}_{b/p}$ by considering the Lyapunov function
\begin{equation} \label{COL-eq:Vfc2}
V_{2} = \frac{1}{2} y^2_{b/p}.
\end{equation}
The time derivative of \eqref{COL-eq:Vfc2} is given by
\begin{align}
\begin{split}
\dot{V}_{2} &= y_{b/p}\dot{y}_{b/p} \\
&\leq -u_{td} \frac{y_{b/p}}{\sqrt{\Delta^2+(y_{b/p}+g)^2}} + (G_1(\cdot)+\tilde{V}_N)y_{b/p} \\
&\leq (G_1(\cdot) + \tilde{V}_N)y_{b/p}
\end{split}
\end{align}
where using the bound on $G_1(\cdot)$ from the paper and Young's inequality we obtain
\begin{align} \label{COL-eq:dVfc2}
\dot{V}_2 &\leq V_2 + \frac{1}{2}\left(\zeta^2(\dot{\gamma}_p(\theta),u_{td})\Vert [\tilde{\psi},\tilde{r},x_{b/p}]^T \Vert^2 + \tilde{V}^2_N \right) \\
&\leq V_2 + \sigma_2(v_r,\tilde{\psi},\tilde{r},\tilde{V}_N,\tilde{V}_T,x_{b/p})
\end{align}
with $\sigma_2(\cdot) \in \mc{K}_{\infty}$. Consequently, if we view the arguments of $\sigma_2(\cdot)$ as input to the $y_{b/p}$ dynamics, then \eqref{COL-eq:dVfc2} satisfies \citet[Corollary 2.11]{angeli1999forward} and hence $\dot{x}_{b/p}$ and $\dot{y}_{b/p}$ are forward complete. Note that the arguments of $\sigma_2(\cdot)$ are all forward complete and therefore fit the definition of an input signal given in \citet{angeli1999forward}. We have now shown forward completeness of \eqref{COL-eq:35a} and \eqref{COL-eq:35c} and since \eqref{COL-eq:35b} is GES is is trivially forward complete. We can therefore claim forward completeness of the entire closed-loop system \eqref{COL-eq:clfull} and the proof of Lemma \ref{COL-lem1} is complete.

\subsection*{Proof of Lemma \ref{COL-lem2}} \label{COL-app:lem2}

Recall the sway velocity dynamics \eqref{COL-eq:35c}:
\[  \dot{v}_r = X(\tilde{u} + u_{rd}) (r_d + \tilde{r}) + Y(u_{rd}+\tilde{u}) v_r ,~~Y(u_{rd}) < 0  \]
Consider the following Lyapunov function candidate:
\begin{align} \label{COL-eq:lyap}
	V_3(v_r) = \frac{1}{2} v^2_r
\end{align}
The derivative of \eqref{COL-eq:lyap} along the solutions of \eqref{COL-eq:35c} is given by
\begin{align} \label{COL-eq:dlyapLem2}
\begin{split}
\dot{V}_3 = &~v_r \dot{v}_r = v_r X(u_{rd} + \tilde{u}) r_d + X(u_{rd} + \tilde{u}) v_r \tilde{r} + Y(u_{rd} + \tilde{u}) v^2_r \\ \leq &~X(u_{rd}) r_d v_r + a_x \tilde{u} r_d v_r + X(u_{rd}) v_r \tilde{r} + a_x \tilde{u} v_r \tilde{r} + a_y \tilde{u} v^2_r + Y(u_{rd}) v^2_r  
\end{split}
\end{align}
where we used the fact that: 
\begin{align} \label{COL-eq:YX}
 Y(u_{r}) = &~a_y u_r + b_y \\
 X(u_r) = &~a_x u_r + b_x    
\end{align}
The term $r_d v_r$ can be bounded as a function of $v_r$ as follows
\begin{align} \label{COL-eq:rdvrlem2}
\begin{split}
r_d v_r = & -\frac{v_r}{C_r} \left[ \kappa(\theta) \left( \frac{u_t \cos(\psi + \beta - \gamma_p(\theta)) + k_\delta x_{b/p} + \hat{V}_{T}}{ 1 - \kappa(\theta) y_{b/p} } \right) \right. \\ &+ \left.  \frac{ Y(u_r) v_r  u_{rd} - \dot{u}_{rd} v_r } { u^2_{rd} + v^2_r } + \frac{ \Delta }{\Delta^2 + \left( y_{b/p} + g \right)^2 } \left[  \dot{\hat{V}}_{N} \frac{ b + \sqrt{b^2-ac} }{-a} \right. \right. \\ &+ \left. \frac{\partial g}{\partial a} \left( 2 \hat{V}_{N} \dot{\hat{V}}_{N} - 2 u_{rd} \dot{u}_{rd} - 2 v_r Y(u_r) v_r \right) + \frac{\partial g}{\partial b} \left( 2 \dot{\hat{V}}_{N} y_{b/p} \right) \right. \\ &+ \left. \left. \left[ 1 + \frac{\partial g}{\partial c} 2 y_{b/p}  + \frac{\partial g}{\partial b} 2 \hat{V}_{N} \right]   \left(  \frac{- u_{td}  y_{b/p} } { \sqrt{\Delta^2 + (y_{b/p} + g)^2 } } + G_1(\cdot)  \right) \right]\right]	\\
\leq &~\frac{1}{C_r} \left|\kappa(\theta)\right| v^2_r \frac{1}{1-\kappa(\theta) y_{b/p}} + F_2 ( \tilde{X}_1,\tilde{X}_2, \Delta, V_T, V_N, u_{rd} ) v^2_r 
\\ &+  F_1 ( \tilde{X}_1,\tilde{X}_2, \Delta, V_T, V_N, u_{rd} ) v_r  
\\ &- \frac{1}{C_r} \left( \frac{ u_{rd}}{ u^2_{rd} + v^2_r } - \frac{2\Delta v_r}{\Delta^2 + (y_{b/p} + g)^2} \frac{\partial g }{\partial a}\right)Y(u_r)v^2_r
\end{split}
\end{align}
where $F_{1,2}(\cdot)$ are continuous functions in their arguments with:
\begin{align} \label{COL-eq:F}
F_2( 0,0,\Delta, V_T, V_N, u_{rd} ) = 0.
\end{align}
When substituting \eqref{COL-eq:rdvrlem2} in \eqref{COL-eq:dlyapLem2} we obtain
\begin{align} \label{COL-eq:Vdotlem2}
\begin{split}
\dot{V}_3 \leq &~X(u_{rd}) F_2 ( \tilde{X}_1,\tilde{X}_2, \Delta, V_T, V_N, u_{rd} ) v^2_r  + \left| \tfrac{ C^*_r - C_r }{ C_r C^*_r } \right|\left( \left| X(u_{rd}) \kappa(\theta) \right| - \left| Y(u_{rd}) \right| \right) v^2_r \\ & + \frac{1}{C^*_r} \left[ \left| X(u_{rd}) \right| \left| \kappa(\theta) \right| \left(1+\frac{y_{b/p}}{1 - \kappa(\theta) y_{b/p}}\right) - \left| Y(u_{rd}) \right| + a_y\tilde{u} \right] v^2_r \\ &+\left(X(u_{rd}) F_1 ( \tilde{X}_1,\tilde{X}_2, \Delta, V_T, V_N, u_{rd}) + a_x \tilde{u} (r_d+\tilde{r}) + X(u_{rd}) \tilde{r}\right)v_r   
\end{split}
\end{align}
where $C^*_r (v_r, y_{b/p} , \Delta , V_N, u_{rd} ) = C_r( v_r, y_{b/p} , \Delta , \hat{V}_N = V_N , u_r = u_{rd} )$. When substituting \eqref{COL-eq:rdvrlem2} in \eqref{COL-eq:dlyapLem2} we have used the fact that
\begin{equation}
\frac{1}{C_r} \left( \frac{ u_{rd}}{ u^2_{rd} + v^2_r } - \frac{2\Delta v_r}{\Delta^2 + (y_{b/p} + g)^2} \frac{\partial g }{\partial a}\right) X(u_r) Y(u_r)v^2_r = \frac{C_r -1}{C_r} Y(u_r)v^2_r.
\end{equation}

\begin{remark}
Note that $C^*_r (v_r, y_{b/p} , \Delta , V_{N}, u_{rd} )$ can be found independently of $y_{b/p}$ and $x_{b/p}$ since the terms in $C_r$ are bounded with respect to these variables. 
\end{remark}

Consequently, on the manifold where $(\tilde{X}_1,\tilde{X}_2)=0$ we have
\begin{align} \label{COL-eq:Vdotmanlem2}
\begin{split}
\dot{V}_3 &\leq \frac{1}{C^*_r}  \left(X_{\max} \left| \kappa(\theta) \right| - Y_{\min}\right)v^2_r + X(u_{rd}) F_1 ( 0,0, \Delta, V_T, V_N, u_{rd})|v_r|
\end{split}
\end{align}
which is bounded as long as 
\begin{align} \label{COL-eq:lem2c1}
X_{\max}\left| \kappa(\theta) \right| - Y_{\min} < 0.
\end{align} 
Hence, satisfaction of \eqref{COL-eq:lem2c1} renders the quadratic term in \eqref{COL-eq:Vdotmanlem2} negative and since the quadratic term is dominant for sufficiently large $v_r$, \eqref{COL-eq:Vdotmanlem2} is negative definite for sufficiently large $v_r$. If $\dot{V}_3$ is negative for sufficiently large $v_r$ this implies that $V_3$ decreases for sufficiently large $v_r$. Since $V_3 = 1/2v^2_r$, a decrease in $V_3$ implies a decrease in $v^2_r$ and by extension in $v_r$. Therefore, $v_r$ cannot increase above a certain value and $v_r$ is bounded near the manifold where $(\tilde{X}_1,\tilde{X}_2)=0$.  

Consequently, close to the manifold where $(\tilde{X}_1,\tilde{X}_2)=0$   the sufficient and necessary condition for local boundedness of \eqref{COL-eq:35c} is the following:
\begin{align}
X_{\max} \left| \kappa(\theta) \right| - Y_{\min} < 0.
\end{align} 
which is satisfied if and only if the condition in Lemma \ref{COL-lem2} is satisfied.

\subsection*{Proof of Lemma \ref{COL-lem3}} \label{COL-app:lem3}

Recall the sway velocity dynamics \eqref{COL-eq:35c}:
\[  \dot{v}_r = X(\tilde{u} + u_{rd}) (r_d + \tilde{r}) + Y(u_{rd}+\tilde{u}) v_r ,~~Y(u_{rd}) < 0  \]
Consider the following Lyapunov function candidate:
\begin{align} \label{COL-eq:lyap2}
	V_3(v_r) = \frac{1}{2} v^2_r
\end{align}
The derivative of \eqref{COL-eq:lyap2} along the solutions of \eqref{COL-eq:35c} is given by
\begin{align} \label{COL-eq:dlyap2}
\begin{split}
\dot{V}_3 = &~v_r \dot{v}_r = v_r X(u_{rd} + \tilde{u}) r_d + X(u_{rd} + \tilde{u}) v_r \tilde{r} + Y(u_{rd} + \tilde{u}) v^2_r \\ \leq &~X(u_{rd}) r_d v_r + a_x \tilde{u} r_d v_r + X(u_{rd}) v_r \tilde{r} + a_x \tilde{u} v_r \tilde{r} + a_y \tilde{u} v^2_r + Y(u_{rd}) v^2_r  
\end{split}
\end{align}
where we used the fact that: 
\begin{align} \label{COL-eq:YX2}
 Y(u_{r}) = &~a_y u_r + b_y \\
 X(u_r) = &~a_x u_r + b_x    
\end{align}
The term $r_d v_r$ is given by:
\begin{align}
\begin{split}
r_d v_r = & - \frac{1}{C_r} v_r \left[ \kappa(\theta) 
\frac{u_t \cos(\psi + \beta - \gamma_p(\theta))}{ 1 - \kappa(\theta) y_{b/p} } + \kappa(\theta) \frac{ k_\delta x_{b/p} + \hat{V}_T}{ 1 - \kappa(\theta) y_{b/p} } \right. 
\\ &+ \left.  \frac{ \Delta \left( b + \sqrt{b^2-ac} \right) }{a \Delta^2 + a \left( y_{b/p} + g \right)^2 }  \left( - k_{x_1} \tilde{x}  \sin(\gamma_p(\theta)) + k_{y_1} \tilde{y} \cos(\gamma_p(\theta)) \right)\right. 
\\ &+ \left. \frac{ \Delta \kappa(\theta)\hat{V}_T\left( b + \sqrt{b^2-ac} \right) }{a \Delta^2 + a \left( y_{b/p} + g \right)^2 } \left( \frac{u_t \cos(\psi + \beta - \gamma_p(\theta))}{ 1 - \kappa(\theta) y_{b/p} } + \frac{ k_\delta x_{b/p} - \hat{V}_T}{ 1 - \kappa(\theta)y_{b/p} }\right) \right. 
\\ &+ \left. \frac{ \Delta \frac{\partial g}{\partial a} 2 \hat{V}_N }{ \Delta^2 + \left( y_{b/p} + g \right)^2 } \left( k_{x_1} \tilde{x}  \sin(\gamma_p(\theta)) - k_{y_1} \tilde{y} \cos(\gamma_p(\theta)) \right) \right.
\\ &- \left. \frac{ \Delta \kappa(\theta) \frac{\partial g}{\partial a} 2 \hat{V}_N\hat{V}_T }{ \Delta^2 + \left( y_{b/p} + g \right)^2 } \left( \frac{u_t \cos(\psi + \beta - \gamma_p(\theta))}{1 - \kappa(\theta) y_{b/p}} + \frac{ k_\delta x_{b/p} + \hat{V}_T}{ 1 - \kappa(\theta)y_{b/p}}\right) \right. 
\\ &- \left. \frac{ \Delta \frac{\partial g}{\partial a} }{ \Delta^2 + \left( y_{b/p} + g \right)^2 } \left(2 u_{rd} \dot{u}_{rd} - 2 v_r Y(u_r) v_r \right) + \frac{ Y(u_r) v_r  u_{rd} - \dot{u}_{rd} v_r }{ u^2_{rd} + v^2_r }\right. 
\\ &+ \left.  \frac{ \Delta \frac{\partial g}{\partial b} 2 y_{b/p} }{\Delta^2 + \left( y_{b/p} + g \right)^2 } \left( k_{x_1} \tilde{x}  \sin(\gamma_p(\theta)) - k_{y_1} \tilde{y} \cos(\gamma_p(\theta)) \right)\right.
\\ &- \left. \frac{ \Delta \kappa(\theta) \frac{\partial g}{\partial b} 2 y_{b/p}\hat{V}_T }{ \Delta^2 + \left( y_{b/p} + g \right)^2 } \left(\frac{u_t \cos(\psi + \beta - \gamma_p(\theta))}{1 - \kappa(\theta) y_{b/p}}+ \frac{ k_\delta x_{b/p} + \hat{V}_T }{ 1 - \kappa(\theta)y_{b/p}} \right)\right. 
\\ &- \left. \phi(\cdot) u_{td}  \frac{ y_{b/p} } { \sqrt{\Delta^2 + (y_{b/p} + g)^2 } } + \phi(\cdot) \tilde{u} \sin(\psi-\gamma_p)\right. 
\\ &+ \left. \phi(\cdot) \left[ 1- \cos(\tilde{\psi}) \right] u_{td} \sin \left( \arctan \left(  \frac{y_{b/p} + g}{\Delta} \right) \right) \right. 
\\ &+ \left. \phi(\cdot) \cos \left( \arctan \left(  \frac{y_{b/p} + g}{\Delta} \right) \right) \sin (\tilde{\psi}) u_{td} \right.\\ &- \left. 2\phi(\cdot) x_{b/p} \kappa(\theta) \left( \frac{u_t \cos(\psi + \beta - \gamma_p(\theta))}{1 - \kappa(\theta) y_{b/p}} +\frac{ k_\delta x_{b/p} + \hat{V}_T}{ 1 - \kappa(\theta) y_{b/p} } \right) \right]
\end{split}
\end{align}
where the function $\phi(y_{b/p},v_r,u_{rd},\hat{V}_N,\Delta)$ is bounded by a constant with respect to $v_r$ and defined as 
\begin{align}
\phi(\cdot) &\triangleq \underbrace{\frac{2 \Delta y_{b/p}}{\Delta^2 + \left( y_{b/p} + g \right)^2 } \frac{\partial g}{\partial c} }_{\phi_1(\cdot)} +  \underbrace{\frac{ \Delta }{\Delta^2 + \left( y_{b/p} + g \right)^2 } }_{\phi_2(\cdot)} + \underbrace{\frac{ 2 \Delta \hat{V}_N }{\Delta^2 + \left( y_{b/p} + g \right)^2 } \frac{\partial g}{\partial b}}_{\phi_3(\cdot)} 	
\end{align}
We can rewrite $r_d v_r$ to obtain
\begin{align}
\begin{split}
r_d v_r = & - \frac{1}{C_r} v_r \left[ \kappa(\theta) 
\frac{ u_t \cos(\psi + \beta - \gamma_p(\theta)) }{ 1 - \kappa(\theta) y_{b/p} }  \right. \\ &- \left. \phi_2(\cdot) u_{td}  \frac{ y_{b/p} + g } { \sqrt{\Delta^2 + (y_{b/p} + g)^2 } } + \phi_2(\cdot) \hat{V}_N  \right.  \\ &+ \left. \phi_2(\cdot) \left[ 1- \cos(\tilde{\psi}) \right] u_{td} \sin \left( \arctan \left(  \frac{y_{b/p} + g}{\Delta} \right) \right)  \right. \\ & +\left. \phi_2(\cdot) \cos \left( \arctan \left(  \frac{y_{b/p} + g}{\Delta} \right) \right) \sin (\tilde{\psi}) u_{td} \right] - \frac{1}{C_r} v_r \Phi_1(\cdot) \\
&- \frac{1}{C_r} \left( \frac{ u_{rd}}{ u^2_{rd} + v^2_r } - \frac{2\Delta v_r}{\Delta^2 + (y_{b/p} + g)^2} \frac{\partial g }{\partial a}\right)Y(u_r)v^2_r
\end{split}
\end{align} 
where $\Phi_1(\cdot)$ collects terms that are bounded with respect to $v_r$ and terms that grow linearly with $v_r$ but vanish when $\tilde{X}_2=0$. The function $\Phi_1(\cdot)$ is defined as
\begin{align}
\begin{split}
\Phi_1(\cdot) \triangleq &~\kappa(\theta) \frac{ k_\delta x_{b/p} - \hat{V}_T}{ 1 - \kappa(\theta) y_{b/p} } - \frac{\dot{u}_{rd} v_r }{ u^2_{rd} + v^2_r } + \frac{2 u_{rd} \dot{u}_{rd}  \Delta  }{ \Delta^2 + \left( y_{b/p} + g \right)^2 } \frac{\partial g}{\partial a} 
\\ &+ \left.  \frac{ \Delta \left( b + \sqrt{b^2-ac} \right) }{a \Delta^2 + a \left( y_{b/p} + g \right)^2 }  \left( - k_{x_1} \tilde{x}  \sin(\gamma_p(\theta)) + k_{y_1} \tilde{y} \cos(\gamma_p(\theta)) \right) \right.
\\ &+ \left.  \frac{ \Delta \kappa(\theta)\hat{V}_T\left( b + \sqrt{b^2-ac} \right) }{a \Delta^2 + a \left( y_{b/p} + g \right)^2 } \left( \frac{u_t \cos(\psi + \beta - \gamma_p(\theta))}{ 1 - \kappa(\theta) y_{b/p} } + \frac{ k_\delta x_{b/p} + \hat{V}_T}{ 1 - \kappa(\theta)y_{b/p} }\right)  \right.  
\\ &+ \left. \frac{ \Delta \frac{\partial g}{\partial a} 2 \hat{V}_N }{ \Delta^2 + \left( y_{b/p} + g \right)^2 } \left( k_{x_1} \tilde{x}  \sin(\gamma_p(\theta)) - k_{y_1} \tilde{y} \cos(\gamma_p(\theta)) \right) \right.  
\\ &- \left. \frac{ \Delta \frac{\partial g}{\partial a} 2 \kappa(\theta) \hat{V}_N\hat{V}_T }{ \Delta^2 + \left( y_{b/p} + g \right)^2 } \left( \frac{u_t \cos(\psi + \beta - \gamma_p(\theta))}{1 - \kappa(\theta) y_{b/p}}+\frac{ k_\delta x_{b/p}+ \hat{V}_T}{ 1 - \kappa(\theta)y_{b/p}}\right) \right.  
\\ &+ \left.  \frac{ \Delta \frac{\partial g}{\partial b} 2 y_{b/p} }{\Delta^2 + \left( y_{b/p} + g \right)^2 } \left( k_{x_1} \tilde{x}  \sin(\gamma_p(\theta)) - k_{y_1} \tilde{y} \cos(\gamma_p(\theta)) \right)\right.  
\\ &- \left. \frac{ \Delta \frac{\partial g}{\partial b} 2 y_{b/p} \kappa(\theta)\hat{V}_T }{ \Delta^2 + \left( y_{b/p} + g \right)^2 } \left( \frac{u_t \cos(\psi + \beta - \gamma_p(\theta))}{1 - \kappa(\theta) y_{b/p}} + \frac{ k_\delta x_{b/p}+ \hat{V}_T }{ 1 - \kappa(\theta)y_{b/p}} \right) \right. 
\\ &- \left. \left( \phi_1(\cdot) + \phi_3(\cdot) \right) u_{td}  \frac{ y_{b/p} } { \sqrt{\Delta^2 + (y_{b/p} + g)^2 } } +  \phi(\cdot) \tilde{u} \sin(\psi-\gamma_p) \right. 
\\ &+ \left.  \left(  \phi_1(\cdot) + \phi_3(\cdot) \right)\left[ 1- \cos(\tilde{\psi}) \right] u_{td} \sin \left( \arctan \left(  \frac{y_{b/p} + g}{\Delta} \right) \right) \right. 
\\ &+ \left. \left( \phi_1(\cdot) + \phi_3(\cdot) \right) \cos \left( \arctan \left(  \frac{y_{b/p} + g}{\Delta} \right) \right) \sin (\tilde{\psi}) u_{td} \right.  
\\ &-  2\phi(\cdot)x_{b/p}\kappa(\theta)\left( \frac{u_t \cos(\psi + \beta - \gamma_p(\theta))}{1 - \kappa(\theta) y_{b/p}} + \frac{ k_\delta x_{b/p} + \hat{V}_T }{ 1 - \kappa(\theta)y_{b/p}} \right)
\end{split}
\end{align}
We now introduce $C^*_r(\cdot)$ as defined in the proof of Lemma \ref{COL-lem2}, so we can rewrite $r_d v_r$ to obtain:
\begin{align}
\begin{split}
r_d v_r = & - \frac{1}{C^*_r} v_r \left[ 
\frac{ \kappa(\theta)u_t \cos(\psi + \beta - \gamma_p(\theta)) }{ 1 - \kappa(\theta) y_{b/p} } - \right.  \\ & \left. \phi_2(\cdot) u_{td}  \frac{ y_{b/p} + g } { \sqrt{\Delta^2 + (y_{b/p} + g)^2 } } + \right.  \\ & \left. \phi_2(\cdot) \left[ 1- \cos(\tilde{\psi}) \right] u_{td} \sin \left( \arctan \left(  \frac{y_{b/p} + g}{\Delta} \right) \right) + \right.  \\ & \left. \phi_2(\cdot) \cos \left( \arctan \left(  \frac{y_{b/p} + g}{\Delta} \right) \right) \sin (\tilde{\psi}) u_{td} \right]  - \frac{1}{C_r} v_r \Phi_2(\cdot)\\
&- \frac{1}{C_r} \left( \frac{ u_{rd}}{ u^2_{rd} + v^2_r } - \frac{2\Delta v_r}{\Delta^2 + (y_{b/p} + g)^2} \frac{\partial g }{\partial a}\right)Y(u_r)v^2_r 	
\end{split}
\end{align} 
where $\Phi_2(\cdot)$ collects terms that are bounded with respect to $v_r$ and terms that grow linearly with $v_r$ but vanish when $\tilde{X}_2=0$. The function $\Phi_2(\cdot)$ is defined as
\begin{align}
\begin{split}
	\Phi_2(\cdot) \triangleq &~\Phi_1(\cdot) + \frac{C^*_r - C_r}{C^*_r} \left[ \phi_2(\cdot) \left[ 1- \cos(\tilde{\psi}) \right] u_{td} \sin \left( \arctan \left(  \frac{y_{b/p} + g}{\Delta} \right) \right) \right. \\ &+ \left. \frac{\kappa(\theta)u_t \cos(\psi + \beta - \gamma_p(\theta)) }{ 1 - \kappa(\theta) y_{b/p} } - \frac{\phi_2(\cdot) u_{td}\left( y_{b/p} + g\right) } { \sqrt{\Delta^2 + (y_{b/p} + g)^2 } }  \right. \\ &+  \phi_2(\cdot) \cos \left( \arctan \left(  \frac{y_{b/p} + g}{\Delta} \right) \right) \sin (\tilde{\psi}) u_{td} \\ &- \left.\left( \frac{ u_{rd}}{ u^2_{rd} + v^2_r } - \frac{2\Delta v_r}{\Delta^2 + (y_{b/p} + g)^2} \frac{\partial g }{\partial a}\right)Y(u_r)v_r \right] + \phi_2(\cdot) \hat{V}_N
\end{split}
\end{align}
Considering the above we derive the following upper bound for $r_d v_r$:
\begin{align}
r_d v_r \leq & \left| \frac{1}{C^*_r} v_r \right| \left[  
\frac{ \left|\kappa(\theta)  \right|u_t }{ 1 - \kappa(\theta) y_{b/p} } + 4  \left|\phi_2(\cdot) \right| u_{td} \right]  - \frac{1}{C_r} v_r \Phi_2(\cdot)\\
&- \frac{1}{C_r} \left( \frac{ u_{rd}}{ u^2_{rd} + v^2_r } - \frac{2\Delta v_r}{\Delta^2 + (y_{b/p} + g)^2} \frac{\partial g }{\partial a}\right)Y(u_r)v^2_r 	
\end{align} 
Using the fact that: $ u_t \leq \left|u_r\right| + \left|v_r\right| $, we obtain:
\begin{align} \label{COL-eq:rdvrlem2fi}
\begin{split}
r_d v_r \leq & \left| \frac{v_r}{C^*_r}  \right| \left[  
\frac{\left|\kappa(\theta)\right|\left( \left|u_r\right| + \left|v_r\right|\right) }{ 1 - \kappa(\theta) y_{b/p} } + 4  \left| \phi_2(\cdot) \right| \left|u_{rd}\right| + 4  \left| \phi_2(\cdot) \right| \left|v_r\right| \right] - \frac{v_r}{C_r}  \Phi_2(\cdot) \\
&- \frac{1}{C_r} \left( \frac{ u_{rd}}{ u^2_{rd} + v^2_r } - \frac{2\Delta v_r}{\Delta^2 + (y_{b/p} + g)^2} \frac{\partial g }{\partial a}\right)Y(u_r)v^2_r 	
\\ \leq &	\left| \frac{1}{C^*_r}  \right| \frac{\left|\kappa(\theta)  \right|v^2_r}{1-\kappa(\theta) y_{b/p}} + 4  \left| \frac{1}{C^*_r}  \right| \left| \phi_2(\cdot) \right| v^2_r + \Phi_3 (\cdot)\\
&- \frac{1}{C_r} \left( \frac{ u_{rd}}{ u^2_{rd} + v^2_r } - \frac{2\Delta v_r}{\Delta^2 + (y_{b/p} + g)^2} \frac{\partial g }{\partial a}\right)Y(u_r)v^2_r 	
\end{split}
\end{align}
where $\Phi_3$ collects the terms that grow linear in $v_r$ and terms that grow quadratically in $v_r$ but vanish when $\tilde{X}_2 = 0$. The function $\Phi_3$ is defined as
\begin{align}
	\Phi_3 (\cdot) \triangleq & \left| \frac{1}{C^*_r} \right| \frac{\left|\kappa(\theta)  \right|  \left| v_r u_r \right| }{1 - \kappa(\theta) y_{b/p}} + \left|\frac{1}{C^*_r} \right|
	\left| v_r u_{rd} \right| \left| \phi_2 (\cdot) \right| - \frac{1}{C_r} v_r \Phi_2(\cdot)
\end{align}
Observing the definition of $\Phi_3(\cdot)$ one can easily conclude the existence of three continuous positive functions $ F_{0,2} ( \tilde{X}_1, \tilde{X}_2 , u_{rd} , \dot{u}_{rd} , V_T , V_N, \Delta ) $ which are bounded since the vector $ [ \tilde{X}^T_2, u_{rd}, \dot{u}_{rd}, V_T, V_N, \Delta ]^T $ is bounded, and where 
\[ 
F_{2} ( \tilde{X}_1, \tilde{X}_2 = 0 , u_{rd} , \dot{u}_{rd} , V_T , V_N, \Delta ) = 0, 
\] 
such that:
\begin{align}
\Phi_3(\cdot) \leq &~F_2(\cdot) v^2_r + F_1(\cdot) v_r + F_0(\cdot)
\end{align}
Consequently, when we substitute \eqref{COL-eq:rdvrlem2fi} in \eqref{COL-eq:dlyap2} obtain:
\begin{align} \label{COL-eq:dlyaplem3}
\begin{split}
\dot{V}_3 = v_r \dot{v}_r \leq & \left| X(u_{rd}) \right| \left[ \left| \frac{1}{C^*_r}  \right| 
\frac{\left|\kappa(\theta) \right|v^2_r}{1-\kappa(\theta) y_{b/p}} + 4  \left| \frac{1}{C^*_r}  \right| \left| \phi_2(\cdot) \right| v^2_r + \Phi_3 (\cdot) \right]  \\ & + a_x \tilde{u} r_d v_r + X(u_{rd}) v_r \tilde{r} + a_x \tilde{u} v_r \tilde{r} +  a_y \tilde{u} v^2_r + Y(u_{rd}) v^2_r \\
&- \frac{1}{C_r} \left( \frac{ u_{rd}}{ u^2_{rd} + v^2_r } - \frac{2\Delta v_r}{\Delta^2 + (y_{b/p} + g)^2} \frac{\partial g }{\partial a}\right)X(u_{rd}Y(u_{rd})v^2_r 	
\\ \leq & \left| \frac{1}{C^*_r}  \right| \left[ \frac{ X_{\max} \kappa_{\max}}{1-\kappa(\theta) y_{b/p}} + 4 X_{\max} \left| \phi_2(\cdot) \right| - Y_{\min} \right] v^2_r   \\ &+  \left|X(u_{rd})\right| \left|\Phi_3 (\cdot) \right| +  a_x \tilde{u} r_d v_r + X(u_{rd}) v_r \tilde{r} + a_x \tilde{u} v_r \tilde{r} + a_y \tilde{u} v^2_r 
\end{split}
\end{align}
To have boundedness of $v_r$ for small values of $\tilde{X}_2$ we have to satisfy the following inequality:
\begin{align} \label{COL-eq:condlem3}
\frac{ X_{\max} \kappa_{\max}}{1-\kappa(\theta) y_{b/p}} + 4 X_{\max} \left| \phi_2(\cdot) \right| -  Y_{\min} < 0
\end{align}
such that the quadratic term in \eqref{COL-eq:dlyaplem3} is negative. Using \eqref{COL-eq:sigma} we need to choose $\Delta$, such that:
\begin{align}
  \left| \phi_2(\cdot) \right| < \frac{ \left[  Y_{\min} -X_{\max} \kappa_{\max}\frac{ 1 }{ \sigma } \right] }{ 4 X_{\max} } > 0,
\end{align}
since $\left| \phi_2(\cdot) \right| \leq \frac{1}{\Delta} $, we can take $ \Delta > \frac{ 4 X_{\max} }{ \left[  Y_{\min} - X_{\max} \kappa_{\max} \frac{ 1 }{ \sigma } \right] }  $ such that \eqref{COL-eq:condlem3} holds. Consequently, near the manifold $\tilde{X}_2 = 0$ it holds that \eqref{COL-eq:dlyaplem3} is negative definite for sufficiently large $v_r$. If $\dot{V}_3$ is negative for sufficiently large $v_r$ this implies that $V_3$ decreases for sufficiently large $v_r$. Since $V_3 = 1/2v^2_r$, a decrease in $V_3$ implies a decrease in $v^2_r$ and by extension in $v_r$. Consequently, $v_r$ cannot increase above a certain value and $v_r$ is bounded near $\tilde{X}_2 = 0$.

\bibliography{biblio}

\bibliographystyle{plainnat}

\end{document}